\documentclass{sig-alternate}
\usepackage{graphicx,latexsym,amsmath,amssymb,mathrsfs,bbm,amscd,mathtools,subfigure,url,verbatim}

\usepackage[all]{xy}
\usepackage{youngtab}
\usepackage{paralist}
\usepackage{tikz}
\usetikzlibrary{decorations,decorations.pathmorphing}
\usetikzlibrary{matrix,calc,backgrounds,arrows,shapes,patterns}
\usepackage{extpfeil}
\usepackage{textcomp}

\newdef{definition}{Definition}
\newtheorem{theorem}[definition]{Theorem}
\newtheorem{claim}[definition]{Claim}
\newtheorem{proposition}[definition]{Proposition}
\newtheorem{corollary}[definition]{Corollary}
\newtheorem{lemma}[definition]{Lemma}
\newtheorem{questions}[definition]{Questions}
\newtheorem{problem}[definition]{Problem}
\newtheorem{exampleI}[definition]{Example}
\newtheorem{remarkI}[definition]{Remark}
\newenvironment{remark}{\begin{remarkI}\normalfont}{\end{remarkI}}
\newenvironment{example}{\begin{exampleI}\normalfont}{\end{exampleI}}

\numberwithin{equation}{section}
\numberwithin{definition}{section}

\newcommand{\fref}[1]{\ref{#1}}

\newcommand{\nopar}{}

\tikzset{obsdesEdgeI/.style={}}
\tikzset{obsdesEdgeII/.style={dashed}}
\tikzset{obsdesEdgeIII/.style={thick, dotted}}
\tikzset{obsdesnode/.style={
        circle,
        draw=black,
        fill=black,
        inner ysep=0cm,
        inner xsep=0cm,
        minimum width = 0.4em,
        minimum height = 0.4em
}
}

\providecommand*\coloneqq{\mathrel{\vcentcolon\mkern-1.2mu=}}

\newcommand{\ComCla}[1]{\ensuremath{\textup{\textbf{\textsf{#1}}}}}

\newcommand{\NP}{\ComCla{NP}}

\newcommand{\sharpP}{\ComCla{\#P}}

\newcommand{\wght}{z}
\renewcommand{\det}{\textup{det}}

\newcommand{\docc}{\textup{\textsf{docc}}}
\newcommand{\affine}{V}
\newcommand{\per}{\textup{per}}

\DeclareMathAlphabet\mathpzc{OT1}{pzc}{m}{it}
\newcommand{\easy}{\mathchoice{\scalebox{1.1}{\ensuremath{\mathpzc{c}}}}{\scalebox{1.1}{\ensuremath{\mathpzc{c}}}}{\scalebox{1.1}{\ensuremath{\scriptstyle\mathpzc{c}}}}{\scalebox{1.1}{\ensuremath{\scriptscriptstyle\mathpzc{c}}}}}
\newcommand{\hard}{\mathchoice{\scalebox{1.1}{\ensuremath{\mathpzc{h}}}}{\scalebox{1.1}{\ensuremath{\mathpzc{h}}}}{\scalebox{1.1}{\ensuremath{\scriptstyle\mathpzc{h}}}}{\scalebox{1.1}{\ensuremath{\scriptscriptstyle\mathpzc{h}}}}}
\newcommand{\Tup}{O_m}
\newcommand{\Tups}{S}
\newcommand{\tup}{s}
\newcommand{\tri}[1]{t_{#1}}
\newcommand{\Tset}{\mathscr{T}}
\newcommand{\omegaX}{X}
\newcommand{\omegaXM}{\mathcal{X}}
\newcommand{\truthlabel}{x}
\newcommand{\elll}{l}

\newcommand{\recta}[2]{#1\mathord{\times}#2}
\newcommand{\pl}[3]{p_{#1}({#2[#3]})}
\newcommand{\singlekron}[1]{k({#1})}
\newcommand{\symmkron}[2]{\mathchoice%
{\mathsf{sk}{\big({#2};\big(#1\big)^2\big)}}%
{\mathsf{sk}{\big({#2};\big(#1\big)^2\big)}}%
{\mathsf{sk}{({#2};(#1)^2)}}%
{\mathsf{sk}{({#2};(#1)^2)}}%
}
\newcommand{\yv}{y}

\newcommand{\eval}{\mathsf{eval}}
\newcommand{\symeval}{\mathsf{seval}}
\newcommand{\dotcup}{\mathrel{\dot{\cup}}}

\DeclareMathOperator{\mult}{mult}

\newcommand{\Weyl}[1]{\{{#1}\}}
\newcommand{\Specht}[1]{[{#1}]}

\newcommand\dash{\nobreakdash-\hspace{0pt}}

\newcommand{\Sym}{\mathsf{Sym}}
\newcommand{\GL}{\mathsf{GL}}
\newcommand{\tensor}{\smash{\textstyle\bigotimes}}

\newcommand{\IC}{\ensuremath{\mathbb{C}}}

\newcommand{\IZ}{\ensuremath{\mathbb{Z}}}

\newcommand{\aS}{\ensuremath{\mathsf{S}}}

\newcommand{\sV}{\ensuremath{\mathscr{V}}}
\newcommand{\sW}{\ensuremath{\mathscr{W}}}
\newcommand{\bra}{\langle}
\newcommand{\ket}{\rangle}

\DeclareMathOperator{\id}{id}
\newcommand{\HWV}{\mathsf{HWV}}
\DeclareMathAlphabet\matheulerscript{U}{eus}{m}{n}
\newcommand{\MaMu}{\ensuremath{\matheulerscript{M}}}
\newcommand{\Er}[1]{\ensuremath{\matheulerscript{E}_{#1}}}
\newcommand{\coord}[1]{\IC[#1]}
\newcommand{\la}{\lambda}
\DeclareMathOperator{\diag}{diag}

\newcommand{\proj}{\mathcal{P}}
\newcommand{\Hy}{\mathcal{H}}
\newcommand{\setp}{\Lambda}

\newcommand{\transpose}[1]{{{^t\!}{#1}}}

\newcommand{\br}{\ensuremath{\underline{R}}}
\newcommand{\ot}{\otimes}
\newcommand{\C}{\ensuremath{\mathbb{C}}}
\newcommand{\N}{\ensuremath{\mathbb{N}}}
\newcommand{\ol}[1]{{\overline{#1}}}

\newcommand{\partinto}[1][]{\smash{\mathord{\mathchoice{%
  \xymatrix@=0.4em@1{%
  \ar@{|-}[rr]_-*--{\scriptstyle #1}
  &*{\phantom{\scriptstyle{#1}}}&}
}{
  \xymatrix@=0.25em@1{%
  \ar@{|-}[rr]_-*--{\scriptstyle #1}
  &*{\phantom{\scriptstyle{#1}}}&}
}{
  \xymatrix@=0.2em@1{%
  \ar@{|-}[rr]_-*--{\scriptscriptstyle #1}
  &*{\phantom{\scriptscriptstyle{#1}}}&}
}{}}}}
\newcommand{\partintonosmash}[1][]{\mathord{\mathchoice{%
  \xymatrix@=0.4em@1{%
  \ar@{|-}[rr]_-*--{\scriptstyle #1}
  &*{\phantom{\scriptstyle{#1}}}&}
}{
  \xymatrix@=0.25em@1{%
  \ar@{|-}[rr]_-*--{\scriptstyle #1}
  &*{\phantom{\scriptstyle{#1}}}&}
}{
  \xymatrix@=0.2em@1{%
  \ar@{|-}[rr]_-*--{\scriptscriptstyle #1}
  &*{\phantom{\scriptscriptstyle{#1}}}&}
}{}}}
\newcommand{\partintostar}[1][]{\smash{\mathord{\mathchoice{%
  \xymatrix@=0.4em@1{%
  \ar@{|-}[rr]_-*--{\scriptstyle #1}^-*--{\scriptstyle \ast}
  &*{\phantom{\scriptstyle{#1}}}&}
}{
  \xymatrix@=0.25em@1{%
  \ar@{|-}[rr]_-*--{\scriptstyle #1}^-*--{\scriptstyle \ast}
  &*{\phantom{\scriptstyle{#1}}}&}
}{
  \xymatrix@=0.2em@1{%
  \ar@{|-}[rr]_-*--{\scriptscriptstyle #1}^-*--{\scriptstyle \ast}
  &*{\phantom{\scriptscriptstyle{#1}}}&}
}{}}}}
\newcommand{\partintostarnosmash}[1][]{\mathord{\mathchoice{%
  \xymatrix@=0.4em@1{%
  \ar@{|-}[rr]_-*--{\scriptstyle #1}^-*--{\scriptstyle \ast}
  &*{\phantom{\scriptstyle{#1}}}&}
}{
  \xymatrix@=0.25em@1{%
  \ar@{|-}[rr]_-*--{\scriptstyle #1}^-*--{\scriptstyle \ast}
  &*{\phantom{\scriptstyle{#1}}}&}
}{
  \xymatrix@=0.2em@1{%
  \ar@{|-}[rr]_-*--{\scriptscriptstyle #1}^-*--{\scriptstyle \ast}
  &*{\phantom{\scriptscriptstyle{#1}}}&}
}{}}}

\begin{document}
%
\conferenceinfo{STOC'13,}{June 1--4, 2013, Palo Alto, California, USA.}
\CopyrightYear{2013}
\crdata{978-1-4503-2029-0/13/06}
\clubpenalty=10000
\widowpenalty=10000

\title{Explicit Lower Bounds via Geometric Complexity Theory}
%
%
%
%
%

\numberofauthors{2} 
%
\author{
%
%
\alignauthor
Peter B\"urgisser\titlenote{partially supported by DFG-grant BU 1371/3-2}\\
       \affaddr{Institute of Mathematics}\\
       \affaddr{University of Paderborn}\\
       \affaddr{D-33098 Paderborn, Germany}\\
       \email{pbuerg@math.upb.de}
\alignauthor
Christian Ikenmeyer\titlenote{partially supported by DFG-grant BU 1371/3-2}\\
       \affaddr{Institute of Mathematics}\\
       \affaddr{University of Paderborn}\\
       \affaddr{D-33098 Paderborn, Germany}\\
       \email{ciken@math.upb.de}
}
\date{30 October 2012}

\maketitle
\begin{abstract}
We prove the lower bound 
$\underline R (\MaMu_m) \geq \frac 3 2 m^2-2$
on the border rank of $m\times m$ matrix multiplication
by exhibiting explicit representation theoretic (occurence) 
obstructions in the sense the 
geometric complexity theory (GCT) program. 
While this bound is weaker than the one recently 
obtained by Landsberg and Ottaviani,  
these are the first significant lower bounds obtained 
within the GCT program.
Behind the proof is an explicit description of  
the highest weight vectors in $\Sym^d\tensor^3 (\C^n)^*$ 
in terms of combinatorial objects, called obstruction designs.
This description results from analyzing  
the process of polarization and Schur-Weyl duality. 
\end{abstract}





\subsection*{Acknowledgments}

We are grateful to Matthias Christandl, Jon Hauenstein, Jesko H\"uttenhain, J.M.\ Landsberg, 
and Michael Walter for important discussions.
We also thank Hang Guo, Stefan Mengel, and Tyson Williams for a valuable discussion on the coloring problem. 

\sloppy

\section{Introduction}

The complexity of matrix multiplication is captured by the rank of the matrix multiplication
tensor, a quantity that, despite intense research efforts, is little understood.
Strassen~\cite{stra:87} already observed that the closely related notion of border rank
has a natural formulation as a specific orbit closure problem.
The work~\cite{BI:10}  applied and further developed the collection of ideas from 
Mulmuley and Sohoni~\cite{gct1,gct2} to the tensor framework, which is 
simpler than the one for permanent versus determinant. 
However, the lower bound obtained in~\cite{BI:10} for the border rank $\br(\MaMu_m)$ of the 
$m\times m$ matrix multiplication tensor $\MaMu_m$ is ridiculously small. 
In this work, we considerably improve this bound and obtain 
the first significant lower bounds obtained within the GCT program.

In a first step, by analyzing the process of polarization and Schur-Weyl duality, 
we arrive at an explicit description of a system of generators~$f_\Hy$
of the spaces of highest weight vectors in 
$\Sym^d\tensor^3 (\IC^n)^*$ in terms of combinatorial objects~$\Hy$, 
called obstruction designs (cf.\ Theorem~\ref{th:OD}). 
We define the chromatic index $\chi'(\Hy)$ of obstruction designs
and prove that $f_\Hy(w)=0$ for all tensors $w\in\tensor^3\C^n$ having border rank 
less than $\chi'(\Hy)$ (Proposition~\ref{pro:chrom-vanish}). 
Our lower bound on the border rank of matrix multiplication 
results from choosing a particular family $(\Hy_m)$ of obstruction designs of 
chromatic index roughly $\frac32 m^2$ with the property that 
$f_{\Hy_m}$ does not vanish on the orbit of the tensor $\MaMu_m$ of $m$ by $m$ matrix 
multiplication. (Proving the nonvanishing is the technically most involved 
part of this paper, cf.\ Lemma~\ref{le:NEU}).  
We also show that, asymptotically, our lower bound is the best that can be 
obtained by applying Proposition~\ref{pro:chrom-vanish} 
(i.e., arguing via the chromatic index), 
provided a conjecture due to Alon and Kim~\cite{AK:97} is true. 

Evaluating $f_\Hy$, or testing whether $f_\Hy$ equals the zero polynomial, 
are challenging problems. It would be interesting to analyzing their complexity.

Our lower bound on the border rank of $\MaMu_m$ is slightly below the one by 
Strassen~\cite{stra:83-2} and Lickteig~\cite{Lic:84}, 
and also weaker than the very recent improvement by Landsberg and Ottaviani~\cite{LO:11}. 
We note that the recent preprint by Grigoriev et al.~\cite{GMP:12} also uses 
representation theory for proving lower bounds on border rank of $\MaMu_m$.  
(However, the lower bounds in~\cite{GMP:12} are substantially worse than 
the ones by Strassen and Lickteig.)

The main message of our paper is that 
significant lower bounds can be obtained with geometric complexity theory (GCT). 
As a further evidence for this, we note that recently, in collaboration with  
Jon Hauenstein and J.M.\ Landsberg, we managed 
to prove $\br(\MaMu_2) = 7$ using an explicit construction of 
highest weight vectors of weight~$\la = (5,5,5,5)^3$ 
and relying on computer calculations.
This is remarkable, since this 
was a long\dash standing open problem since the 70s, 
which was only settled in 2005 by
Landsberg \cite{Lan:05} using very different methods.

As a further contribution, we add to the discussion on the feasibility of the GCT approach 
by pointing out that in a modification of the approach, proving lower bounds is actually 
{\em equivalent} to providing the existence of obstructions (in the sense of highest weight vectors 
instead of just highest weights), cf.~Proposition~\ref{pro:obshwv}. 

This work 
contains results from the PhD thesis of the second author~\cite{ike:12b}.

\section{Orbit Closure Problems}
\label{sec:ocp}

\subsection{Border Rank}

Consider $W:=\tensor^3\C^{m^2}$. 
The {\em rank} $R(w)$ of  a tensor $w\in W$ 
is defined as the minimum~$r\in\N$
such that $w$~can be written as a sum of $r$ tensors of the 
form $w^{(1)}\ot w^{(2)}\ot w^{(3)}$ with $w^{(i)}\in\C^n$. 
Strassen proved~\cite{Str:73} that, up to a factor of two, $R(w)$ 
equals the minimum number of nonscalar multiplications sufficient 
for evaluating the bilinear map 
$(\C^{m^2})^*\times (\C^{m^2})^*\to\C^{m^2}$ corresponding to $w$.  
The {\em border rank} $\br(w)$  of a tensor $w\in W$ is defined as the
smallest $r\in \N$ such that $w$ can be obtained as the limit of a sequence
$w_k\in W$ with tensor rank $R(w_k)\le r$ for all $k$. 
Border rank is a natural mathematical notion that has played an important role 
in the discovery of fast algorithms for matrix multiplication, see~\cite[Ch.~15]{ACT}.

Now let $n\ge m^2$ and think of $W=\tensor^3\IC^{m^2}$ as embedded in 
$V:=\tensor^3\C^n$ via an embedding $\IC^{m^2} \subseteq \IC^n$. The group $G:=\GL_n^3$ acts on~$V$ via  
$(g_1,g_2,g_3)(w^{(1)}\ot w^{(2)}\ot w^{(3)}) := g_1(w^{(1)})\ot g_2(w^{(2)})\ot g_3(w^{(3)})$. 
We shall denote by 
$Gw:=\{g w \mid g \in G\}$ 
the {\em orbit} of~$v$ and call its closure $\ol{Gw}$ 
with respect to the euclidean topology 
the {\em orbit closure} of~$w$.
 
It will be convenient to use Dirac's bra-ket notation. So 
$|i\ket$,  $1 \leq i \leq n$, denotes the standard basis of $\C^n$ 
and $\bra i|$ denotes its dual basis. Further, 
$|ijk\ket$ is a short hand for $|i\ket\ot|j\ket\ot|k\ket\in V$.
We call 
$\Er n \coloneqq \sum_{i=1}^n |iii\ket \in V$
the \emph{$n$-th unit tensor}.

Suppose that $\br(w)\ge m$ to avoid trivial cases. 
Then it is easy to see that
$\br(w)\le n$ iff $w \in \ol{G{\Er n}}$, cf.~\cite{stra:87}.

The tensor corresponding to the $m \times m$ matrix multiplication map 
can be succinctly written as 
\begin{align}
\label{eq:defmamu}
 \MaMu_m \coloneqq \sum_{i,j,l=1}^m |(i,j)(j,l)(l,i)\ket \in \tensor^3 \IC^{m\times m}.
\end{align}

\subsection{Approximate Determinantal Complexity}

We switch now the scenario and take 
$V:= \Sym^n \C^{n^2}$, which is  the 
homogeneous part of degree $n$ of the polynomial ring $\C[X_1,\ldots,X_{n^2}]$.
The determinant $\det_n$ of an $n\times n$ matrix in these variables is an element of $V$.
The group $G:=\GL_{n^2}$ acts on~$V$ by linear substitution.
Further, let $m < n$, and put $z:=X_{m^2+1}$, $W:= \Sym^m \C^{m^2}$. 
We define the \emph{determinantal orbit closure complexity} $\docc(f)$ 
of~$f\in W$ as the minimal $n$ such that
$ z^{n-m}f \in \ol{G \det_n}$.

In \cite{gct1} Mulmuley and Sohoni conjectured the following:
\begin{equation}\label{MS-conj}
\mbox{$\docc(\per_m)$ is not polynomially bounded in~$m$.}
\end{equation}
Here $\per_m\in W$ denotes the permanent of the 
$m\times m$ matrix in the variables $X_1,\ldots,X_{m^2}$.  

An affirmative answer to this conjecture implies that 
$\det_n$ cannot be computed by weakly skew circuits of size polynomial in $m$,
(cf.~\cite{BLMW:11}), which is a version of  Valiant's Conjecture~\cite{Val:79b}.

\subsection{Unifying Notation}
\label{subsec:unifnot}

The {\em tensor scenario} and the {\em polynomial scenario}
discussed before have much in common 
and we strive to treat both situations simultaneously. 
Hence for fixed $n$ and $m$ we want to use the notations 
summarized in the following table. 
\[\begin{array}{l||l|l}
\text{notation} & \text{determinantal} & \text{border rank}\\
& \text{complexity} & (n \geq m^2)\\
  & (n \geq m+1) & \\
\hline
G & \GL_{n^2} & \GL_{n} \times \GL_{n} \times \GL_{n}\\
V & \Sym^n \C^{n^2} & \tensor^3 \C^n\\
\eta = \dim V & \binom{n^2+n-1}{n}  & n^3\\
W \subseteq V & \Sym^m \C^{m^2} & \tensor^3 \C^{m^2}\\
\hard\coloneqq\hard_{m,n} \in W & z^{n-m}\per_m & \MaMu_m \\
\easy\coloneqq\easy_n \in V & \det_n & \Er n
\end{array}\]
The symbol~$\hard$ stands for the hard problem for which we want to prove lower bounds
and the orbit closure $\overline{G\easy_n}$ is exactly the set of all elements in~$V$ with complexity at most~$n$.
In both scenarios, for a given~$m$, we try to find~$n$
as large as possible such that \[\hard_{m,n} \notin \overline{G\easy_n}.\]
Since the orbit closure is the smallest closed set containing the orbit,
this is equivalent to proving
$\ol{G \hard_{m,n}} \not\subseteq \overline{G\easy_n}$. 
If we want to treat $G\easy$ and $G\hard$ simultaneously, we just write $Gv$.

\section{The Flip via Obstructions}  

Let $V \simeq \IC^\eta$ and $v\in V$ in one of the two scenarios above. 
We write $\coord{V} := \IC[T_1,\ldots,T_\eta]$ for the ring of
polynomial functions on $V$.
It is a fundamental fact from algebraic geometry that the 
orbit closures $\ol{Gv}$ (defined via the euclidean topology) 
are in fact Zariski closed, i.e., zero sets of polynomials on $V$
(cf.~\cite[AI.7.2]{Kra:85}). This immediately implies the following observation.

\begin{proposition}
\label{cor:polyobstrexist}
Let $\hard \in \affine$.
If $\hard \notin \overline{G\easy}$, 
then there exists a polynomial
$f \in \coord{\affine}$ that vanishes on $\overline{G\easy}$ but not on~$\hard$.
\end{proposition}

We call such polynomials~$f$ that separate $\hard$ from $\overline{G\easy}$ 
\emph{polynomial obstructions}. 
By Proposition~\ref{cor:polyobstrexist}, they are 
guaranteed to exist if $\hard \notin \overline{G\easy}$. 
We want to investigate 
whether there are ``short encodings'' of polynomial obstructions~$f$
and whether there are ``short proofs'' that $f$ is an obstruction. 
Representation theory provides a natural framework to address these questions.



\subsection{Highest Weight Vectors}
We recall some facts from representation theory~\cite{gw:09}.
Let $\sV$ be a rational $\GL_n$\dash representation.
For a given $\wght \in \IZ^n$, a \emph{weight vector $f\in\sV$ of weight \wght}
is defined by the following property:
$
 \diag(\alpha) f = \alpha_1^{\wght_1} \alpha_2^{\wght_2} \cdots \alpha_n^{\wght_n} f
$
for all $\alpha \in (\IC^\times)^n$,
where $\diag(\alpha)$ denotes the diagonal matrix with diagonal entries $\diag(\alpha)_{i,i} = \alpha_i$. 

 Let $U_n\subseteq \GL_n$ denote the group of upper triangular matrices with 1s on the main diagonal,
 the so-called \emph{maximal unipotent group}.
 A weight vector $f\in \sV$ that is fixed under the action of~$U_n$, i.e., $\forall u \in U_n: uf = f$,
 is called a \emph{highest weight vector (HWV) of $\sV$}.
 The vector space of HWVs of weight $\la$ is denoted by~$\HWV_\la(\sV)$.
The following is well known. 

\begin{lemma}
\label{lem:highestweightvector}
 Each irreducible rational $\GL_n$\dash representation $\sW$ contains, up to scalar multiples, 
exactly one nonzero HWV~$f$.
The representation $\sW$ is the linear span of the $\GL_n$\dash orbit of $f$.
Two irreducible representations are isomorphic iff the weights of their HWVs coincide.
\end{lemma}

The weight $\la\in\IZ^n$ of a HWV is always nondecreasing. 
It describes the isomorphy type of $\sW$. 
The heighest weight of the dual $\sW^*$ of $\sW$ is given by 
$\la^*:=(-\la_n,\ldots,\la_1)$. 
We denote by $\Weyl\la$ the  irreducible $\GL_n$\dash representation 
with highest weight~$\la$, called \emph{Weyl-module}. 
It is a well known fact that every $\sV$ splits into a direct sum of irreducible 
$\GL_n$-representations.

What has been said for $\GL_n$ extends in a straightforward way to 
representations $\sV$ of the group $\GL_n\times\GL_n\times\GL_n$.
A weight vector $f\in\sV$ of weight $z\in\IZ^n\times\IZ^n\times\IZ^n$ 
satisfies
$$
\mbox{$(\diag(\alpha^{(1)}), \diag(\alpha^{(2)}),\diag(\alpha^{(3)}))f 
 = \prod_{k=1}^3 \prod_{i=1}^n (\alpha_i^{(k)})^{z^{(k)}_i} $} \!f
$$
for all $\alpha^{(k)} \in (\IC^\times)^n$. 
The type of irreducible $\GL_n\times\GL_n\times\GL_n$\dash representations 
is given by triples $\la=(\la^{(1)},\la^{(2)},\la^{(3)})$, 
where $\la^{(k)}$ is a highest weight for $\GL_n$.  
We also write $\la^* := ((\la^{(1)})^*,(\la^{(2)})^*,(\la^{(3)})^*)$. 

\subsection{HWV Obstructions}

We return to our two scenarios.
The action of the group $G$ on $V$ induces an action of $G$ on $\C[V]$ defined by 
$(gf)(x) := f(g^{-1}x)$ for $g\in G$, $f\in\C[V]$, $x\in V$. 
This action respects the degree~$d$ part $\sV=\C[V]_d$. 
Let $I(Gv)=I(\overline{Gv})$ denote the vanishing ideal of the orbit~$Gv$
and let $\coord{\overline{Gv}} := \coord{V}/I(\overline{Gv})$
denote the coordinate ring of $\ol{Gv}$.

The following result shows that when searching for polynomial obstructions, 
we can restrict ourselves to HWVs.

\begin{proposition}
\label{pro:obshwv}
Let $\hard \in \affine$.
If $\hard \notin \overline{G\easy}$, 
then there exists 
some HWV $f_\la \in \coord{V}$ of some weight~$\la$
such that $f_\la$ vanishes on $\overline{G\easy}$, but 
$f_\la(g\hard)\neq 0$ for some $g \in G$.
\end{proposition}
\begin{proof}
 The fact $f(\overline{G\easy})=0$
means that $f$ is contained in the vanishing ideal $I(\overline{G\easy})$.
But $I(\overline{G\easy})$ is a graded $G$\dash representation.
Hence we can write
$
 f = \sum_{d,\la} f_{d,\la},
$
where $f_{d,\la} \in I(\overline{G\easy})_d$ are elements from the
isotypic component of type~$\la$ in
the homogeneous part $I(\overline{G\easy})_d$.
By Lemma~\ref{lem:highestweightvector},
it follows that we can write
$
 f_{d,\la} = \sum_i g_{d,\la,i} f_{d,\la,i},
$
where $g_{d,\la,i} \in G$ and $f_{d,\la,i}$ is a HWV
in $I(\overline{G\easy})_d$ of weight~$\la$.

Let $g \in G$ with $f(g\hard)\neq 0$.
Then $g_{d,\la,i} f_{d,\la,i}(g \hard) \neq 0$ for some $d,\la,i$.
This means $f_{d,\la,i}(g_{d,\la,i}^{-1} g \hard) \neq 0$,
which proves the proposition.
\end{proof}

We call such $f_\la$ a {\em HWV obstruction}
against $\hard \in \overline{G\easy}$.
We will show that some HWVs have a succinct encoding,
which is linear in their degree~$d$. 

\begin{problem}
Can the separation in Proposition~\ref{pro:obshwv} 
always be achieved by some HWV $f_\la$ of degree~$d$ 
polynomially bounded in~$n$?
(We can show an upper bound exponential in~$n$ 
using general results on quantifier elimination over $\C$.)
\end{problem}


An \emph{occurence obstruction} against $\hard \in \overline{G\easy}$,
as introduced by Mulmuley and Sohoni~\cite[Def.~1.2]{gct2}, 
is a highest weight~$\la$ for $G$ such that 
irreducible $G$-representations of type $\la$ do not occur in 
$\coord{\overline{G\easy}}$, 
but some irreducible $G$-representation of type~$\la$ 
does occur in $\coord{\overline{G\hard}}$. 
These properties can be rephrased as follows:
\begin{itemize}
\item {\em All} HWVs in $\C[V]$ of weight $\la$ 
vanish at $\ol{G\easy}$;
\item There exists some HWV~$f_\la$ in $\C[V]$ of
  weight $\la$ that does not vanish on $\ol{G\hard}$.
\end{itemize} 
If $\la$ is an occurence obstruction against
$\hard \in \overline{G\easy}$, then
there exists a HWV obstruction $f_\la$ of weight~$\la$.
But the converse is not true in general, see for instance  
the discussion on Strassen's invariant in~\cite{BI:10}.  
Clearly, if the irreducible represenation corresponding to $\la$ 
occurs in  $\C[V]$ with high multiplicity, then item one 
above is much harder to satisfy for occurence obstructions.

While Proposition~\ref{pro:obshwv} tells us that 
$\hard \not\in \overline{G\easy}$ can, in principle, 
always be proven by exhibiting a HWV obstruction,
it is unclear whether this is also the case for occurence obstructions. 
We state this as an important open problem. 

\begin{problem}
For the scenarios in Subsection~\ref{subsec:unifnot},
if $\hard_{m,n} \notin \overline{G\easy_n}$,
is there an occurence obstruction proving this?
\end{problem}

Mulmuley and Sohoni conjecture that~(\ref{MS-conj}) 
can be proved with occurence obstructions, see \cite[\S3]{gct2}.

\section{Main Results}\label{se:main} 

\subsection{Some Notation}

A \emph{partition} $\la$ is a finite sequence of nonincreasing natural numbers.
The number of nonzero elements in $\la$ is called its \emph{length}~$\ell(\la)$.
We call $|\la| \coloneqq \sum_i \la_i$ the \emph{size} of $\la$.
If $\la$ satisfies $|\la|=d$ and $\ell(d)\leq n$, then we write $\la \partinto[n] d$.
If we do not specify the size, we just write $\la \partinto[n]$,
and if we do not specify the length, we write $\la \partinto d$.

A pictorial description of partitions is given by \emph{Young diagrams},
which are upper-left-justified arrays having $\la_i$ boxes in the $i$th row.
The partitions $\recta \ell k := (k,k,\ldots,k)$ correspond to rectangular Young diagrams with $\ell$ rows and $k$ columns. 
When reflecting a Young diagram~$\la$ at the diagonal from the upper left to the lower right
we get a Young diagram again, which we call the \emph{transposed Young diagram}~$\transpose \la$.
Note that the number of boxes of $\la$ in column $i$ equals $\transpose\la_i \coloneqq (\transpose\la)_i$.

For a triple $\la = (\la^{(1)},\la^{(2)},\la^{(3)})$ of partitions, 
henceforth called {\em partition triple}, 
we use the short notation $\la \partintostar[n]d$ 
to express that $\la^{(k)}\partinto[n]d$ for all $k \in \{1,2,3\}$.

A \emph{set partition} $\setp$ of a set $S$ is a set of subsets of $S$ such that 
for all $s \in S$ there exists exactly one $e_s \in \setp$ with $s \in e_s$. 
If $\mu$ denotes the partition obtained from sorting the multiset $\{|e| : e \in \setp\}$, 
then we call the partition $\transpose\mu$ the \emph{type} of~$\setp$.
(The reason for taking the transpose will become clear soon.) 

\subsection{Obstruction Designs Encoding HWVs}


The following reasonings require some multilinear algebra.

Consider an ordered set $S=[d]:=\{1,2,\ldots,d\}$. 
We interpret a map $J\colon S\to\C^n$ as an $n\times d$-matrix whose columns are 
indexed by the elements of $S$. For a subset $e\subseteq S$ we denote by 
$\det J|_e$ the determinant of the submatrix of $J$ obtained 
by selecting the first $|e|$ rows and the columns indexed by the elements in~$e$. 
We define the evaluation of a partition $\setp$ of the set $S$ at $J$ by 
\begin{equation}\label{eq:def-eval}
 \eval_\setp(J) := \prod_{e\in\setp} \det J|_e .
\end{equation}
Note that $(\C^n)^d\to\C,\,J\mapsto \eval_\setp(J)$ is multilinear 
and therefore defines a linear form $\eval_\setp$ on $\tensor^d \C^n$. 

We define an {\em obstruction design} as 
a subset $\Hy\subseteq [\ell_1]\times [\ell_2]\times [\ell_3]$ of the discrete box 
of side lengths $\ell_1,\ell_2,\ell_3$, respectively. 
The {\em $1$-slices}  of $\Hy$ are defined as the sets 
$e^{(1)}_i:=\{ (i,j,k) \in \Hy\mid j \in [\ell_2], k\in [\ell_3]\}$ for $i\in [\ell_1]$. 
The set $E^{(1)}$ consisting of the $1$-slices of $\Hy$ defines 
a set partition of~$\Hy$ (after omitting possibly empty 1-slices). 
The first marginal distribution of $\Hy$ is the map 
$\mu^{(1)}\colon [\ell_1]\to \N, i \mapsto |e^{(1)}_i|$. 
Similarly, we define the set partition $E^{(2)}$ of $2$-slices of~$\Hy$ 
with its  marginal distribution $\mu^{(2)}$ and 
the set partition $E^{(3)}$ of $3$-slices of~$\Hy$
with its  marginal distribution $\mu^{(3)}$. 
Note that $|e^{(1)}\cap e^{(2)} \cap e^{(3)}| \le 1$ for all 
$(e^{(1)},e^{(2)},e^{(3)}) \in E^{(1)}\times  E^{(2)}\times  E^{(3)}$. 
By a permutation of the sides we may always assume that 
the marginal distributions $\mu^{(k)}$ are monotonically decreasing, i.e., 
partitions of $d:=|\Hy|$. 
Then $\la^{(k)} := \transpose\mu^{(k)}$  is the type of the set partition $E^{(k)}$.  
We call the {\em partition triple} 
$\la=(\la^{(1)},\la^{(2)},\la^{(3)})$ the {\em type} of $\Hy$. 
If all slices contain at most $n$ elements, 
we have 
$\la \partintostar[n]d$.  

By a {\em triple labeling} of $\Hy$ we shall understand 
a map $J\colon\Hy\to (\C^n)^3$. 
After fixing an ordering of $\Hy\simeq [d]$, 
we can define the 
{\em evaluation of the obstruction design $\Hy$ at the triple labeling~$J$} 
by
$$
 \eval_\Hy(J) := \prod_{k=1}^3 \eval_{E^{(k)}} (J^{(k)}) ,
$$
where $J^{(k)}\colon\Hy\to\C^n$ denote the 
components of $J$.
Note that 
$(\C^n)^d\times (\C^n)^d\times (\C^n)^d \to \C,\, J\mapsto \eval_\Hy(J)$
is multilinear. Hence it defines a linear form 
$\tensor^d \tensor^3\C^n \to \C$ 
that we denote by the same symbol: 
$\eval_\Hy( \tensor_{s=1}^d \tensor_{k=1}^3 J^{(k)}(s)) := \eval_\Hy(J)$. 
We symmetrize the linear form $\eval_\Hy$ with respect to the permutations in $S_d$
\begin{eqnarray*}
\lefteqn{\symeval\big(\tensor_{s=1}^d \tensor_{k=1}^3 J^{(k)}(s)\big)} \\ 
  &:=& 
 \frac{1}{d!} \sum_{\pi\in\aS_d} \eval_\Hy \big(\tensor_{s=1}^d \tensor_{k=1}^3 J^{(k)}(\pi(s))\big)
\end{eqnarray*}
obtaining a symmetric multilinear form 
on $\tensor^d \tensor^3\C^n$. 
Then $f_\Hy(w) := \symeval_\Hy(w^{\ot d}) = \eval_\Hy(w^{\ot d})$ defines 
a homogenous polynomial $f_\Hy$ of degree~$d$ on $\tensor^3 \C^n$ 
(restitution and polarization, cf.\ in \cite[Ch.~1.2]{Dol:03}).

More specifically, the polynomial~$f_\Hy$ can be described as follows. 
Suppose that the tensor $w$ is decomposed into distinct rank 1 tensors as 
$w=\sum_{i=1}^r w^{(1)}_i \ot  w^{(2)}_i \ot  w^{(3)}_i$.
We have 
$$
 w^{\ot d} = \sum_{I\colon [d]\to [r]} \tensor_{s=1}^d \tensor_{k=1}^3 w_{I(s)}^{(k)}
$$
Consider the set 
$\Tset:= \{ (w^{(1)}_i, w^{(2)}_i ,w^{(3)}_i) \mid 1 \le i \le r\}$ 
of triples of vectors.
The maps $I\colon [d]\to [r]$ correspond bijectively to the 
triple labelings $J\colon\Hy\to\Tset$ defined by $J^{(k)}(s) := w_{I(s)}^{(k)}$.
Therefore, 
$$
 \eval_\Hy\big(w^{\otimes d}\big) = 
 \sum_{J\colon \Hy\to \Tset} \eval_\Hy\big(\tensor_{s=1}^d \tensor_{k=1}^3 J^{(k)}(s) \big).
$$ 
This implies 
\begin{equation}\label{def:fR}
 f_\Hy(w) = \sum_{J\colon \Hy\to \Tset} \eval_\Hy(J) . 
\end{equation}

By symmetry, $f_\Hy(w)$ does not depend on the chosen ordering of~$\Hy$. 

\begin{theorem}\label{th:OD}
Let $\Hy$ be an obstruction design of type $\la \partintostar[n]d$. 
Then $f_\Hy$ is a highest weight vector of weight $\la^*$ in $\Sym^d\tensor^3(\IC^*)^n$. 
Moreover, if $\Hy$ runs over all obstructions designs of type $\la$, 
the $f_\Hy$ span the space of highest weight vectors $\HWV_{\la^*}(\Sym^d\tensor^3(\IC^*)^n)$. 
\end{theorem}

The proof will be given in Section~\ref{se:explicit-HWVs}. 

\subsection{Chromatic Index of Obstruction Designs}

We describe here a simple combinatorial condition for $f_\Hy$ 
vanishing on all tensors of border rank at most~$r$. 
Let us stress that this condition is sufficient, 
but far from being necessary. 

By a {\em proper coloring} of an obstruction design~$\Hy$ 
with $c$ colors we shall understand a map 
$\sigma\colon\Hy \to [c]$ such that in each slice of~$\Hy$, 
the colors of points are pairwise different. 
The {\em chromatic index} $\chi'(\Hy)$ is defined as the 
least number of colors sufficient for coloring~$\Hy$. 

\begin{proposition}\label{pro:chrom-vanish}
Let $\Hy$ be an obstruction design of type $\la \partintostar[n]d$.
Then we have $f_\Hy(w)=0$ for all tensors $w\in\tensor^3\C^n$ satisfying 
$\br(w)<\chi'(\Hy)$. 
\end{proposition}

\begin{proof}
Suppose that 
$w=\sum_{i=1}^r w^{(1)}_i \ot  w^{(2)}_i \ot  w^{(3)}_i$
and interpret $\Tset$ (defined right before~\eqref{def:fR}) as a set of colors. 
If $r = |\Tset|<\chi'(\Hy)$, then a map $J\colon\Hy\to\Tset$ 
cannot be a proper coloring of~$\Hy$. 
Hence there exists some~$k$ and some slice $e\in E^{(k)}$ 
in which two points get the same color. 
As a consequence, the matrix $J^{(k)}|e$ has a duplicated column 
and hence $\det J^{(k)}|e = 0$. Therefore, 
Equation~\eqref{def:fR} implies that $f_\Hy(w)=0$. 
By continuity it follows that $f_\Hy(v)=0$ for all $v\in\tensor^3\C^n$ 
with $\br(v)\le r$. 
\end{proof}


It is therefore desirable to find obstruction designs with large chromatic index. 
There is a limit though.

\begin{lemma}
We have $\chi'(\Hy) \le 3n -2$ for any obstruction design of type $\la \partintostar[n]d$.
\end{lemma}

\begin{proof}
$\chi'(\Hy)$ equals the chromatic number of the graph~$G$ 
with vertex set $\Hy$, in which two nodes are connected iff they lie in a same slice. 
Each node in this graph has degree at most 
$\Delta = 3(n-1)$, since there are at most $n$ nodes in each slice.
It is well known from graph theory that $1+\Delta$ is an upper bound 
on the chromatic number of $G$. 
\end{proof}

This result shows that $3n-2$ is the best lower bound on border rank that can be shown 
based on Proposition~\ref{pro:chrom-vanish}. 
Unfortunately, the limit seems even smaller. 

An obstruction design $\Hy\subseteq [\ell_1]\times [\ell_2]\times [\ell_3]$ 
can be interpreted as a 3-partite, 3-uniform, 2-simple hypergraph: its set of nodes 
is the disjoint union $[\ell_1]\dot{\cup} [\ell_2]\dot{\cup} [\ell_3]$ and each $(i,j,k)\in\Hy$ 
defines a hyperedge $\{i,j,k\}$. With this view in mind, $\chi'(\Hy)$ turns 
out to be the chromatic index of this hypergraph: indeed we want to color 
the hyperedges in such a way that incident hyperedges get different colors. 
The degree of this hypergraph is the maximum cardinality of slices of $\Hy$, 
which is bounded by~$n$. 
A conjecture due to Alon and Kim~\cite{AK:97} implies that for all $\epsilon >0$,  
there is $n_0$ such that for all $n\ge n_0$ we have 
$\chi'(\Hy) \le (\frac32 + \epsilon) n$ for all $\Hy$ of type $\la \partintostar[n]d$.
Hence, if this conjecture is true,
$\frac32 n +o(n)$ is the best possible lower bound 
on border rank that can be shown based on Proposition~\ref{pro:chrom-vanish}. 

We next show that we can achieve this lower bound for the matrix multiplication tensor. 


\subsection{Lower Bound for Matrix Multiplication}\label{ssec:lbMaMu}

Consider the obstruction design 
$\Hy_\kappa := \{(i,j,k)\in [\kappa+1]^3 \mid i=1 \mbox{ or } j=1 \mbox{ or } k=1 \}$
given by a ``3-dimensional hook'' ($\kappa\in\N$).  
Its type $\la(\kappa)$ is the triple with components three times the hook partition 
$(\kappa+1,1,\ldots,1)\partinto[2\kappa+1] 3\kappa+1$.  
It is obvious that $\chi'(\Hy_\kappa) = 3\kappa+1$. 

In Section~\ref{se:NVatMaMu} we shall prove the following technical result.

\begin{lemma}\label{le:NEU}
There exists a matrix triple $A \in (\GL_{m^2})^3$ such that 
$f_{\Hy_\kappa}(A\MaMu_m) \neq 0$,
where $\kappa := \frac{m^2-1}{2}$ for $m > 1$ odd.
\end{lemma}

Combined with Proposition~\ref{pro:chrom-vanish}, this implies the following result 
if $m$ is odd. 
(The proof where $m$ is even is omitted.)

\begin{theorem}\label{thm:hooklowerbound}
We have $\br(\MaMu_m) \ge \frac{3}{2} m^2-\frac{1}{2}$ 
if $m$ is odd. Moreover,  $\br(\MaMu_m) \ge \frac{3}{2} m^2-2$ 
if $m$ is even. 
\end{theorem}

\begin{remark}
The same proof gives the same lower bound on the \emph{s\dash rank}~\cite{CU:12} 
of the matrix multiplication tensor.
\end{remark}

\begin{remark}
One can prove that $\Hy_\kappa$ is the only obstruction design of type $\la(\kappa)$. 
Put $n:=3\kappa$ and $d:=3\kappa +1$. 
Proposition~\ref{pro:chrom-vanish} implies that 
$f_{\Hy_\kappa}$ vanishes on $\overline{\GL_n^3 \Er n}$. 
Therefore, $\mult_{\la(\kappa)}(\coord{\overline{\GL_n^3 \Er n}}_d) =0$.   
Hence $\la(\kappa)$ is an occurence obstruction against  
$\MaMu_m \in \overline{\GL_n^3 \Er n}$. 
Based on the results in ~\cite{BI:10} 
we can prove the stronger statement 
$\mult_{\la(\kappa)}(\coord{\GL_n^3 \Er n}_d)=0$, 
cf.~\cite[Prop.~8.3.1]{ike:12b}. 
\end{remark}

\subsection{Comments, Examples, Open Questions}

Permutations in $S_{\ell_1}\times S_{\ell_2} \times S_{\ell_3}$ naturally act on 
the discrete cube $[\ell_1]\times [\ell_2]\times [\ell_3]$.
We call two obstruction designs $\Hy_1,\Hy_2\subseteq [\ell_1]\times [\ell_2]\times [\ell_3]$
{\em equivalent} if $\Hy_2$ arises from $\Hy_1$ by applying such a permutation. 
This amounts to permuting slices. 
Note that if $\Hy_1$ and $\Hy_2$ have the same type, 
then we can only permute slices having the same cardinality. 
Let $N(\la)$ denote the number of equivalence classes of obstructions designs of type~$\la$. 
It is clear that $f_{\Hy_1}=f_{\Hy_2}$ if $\Hy_1$ and $\Hy_2$ are equivalent. 

The {\em Kronecker coefficient} $\singlekron \la$ of $\la\partintostar[n]d$ 
can be characterized as the dimension of 
$\HWV_{\la^*}(\Sym^d\tensor^3(\IC^*)^n)$, 
cf.~\cite{BI:10}.  Theorem~\ref{th:OD} therefore implies the following upper bound
on Kronecker coefficients, which appears to be new 
(this is related, but different from~\cite{TCV:98}). 

\begin{corollary}\label{cor:kronest}
We have $\singlekron\la \leq N(\la)$ for $\la\partintostar[n]d$. 
\end{corollary}

Mulmuley~\cite{mulmuley:11} conjectures that deciding $\singlekron\la >0$ 
is possible in polynomial time. This should be 
contrasted with the following result, which follows from~\cite{BDLG:00}.

\begin{proposition}
\label{pro:NPC}
 Given a partition triple $\la\partintostar[n]d$ encoded in unary.
 Then it is $\NP$\dash complete to decide whether there
 exists an obstruction design of type~$\la$.
\end{proposition}

Deciding whether $f_\Hy$ vanishes identically can be difficult even in seemingly 
simple situations. 

\begin{example}\label{ex:AlonTarsi}
Let $\Hy_n:= \{(i,j,n(i-1)+j)\} \subseteq [n]\times [n] \times [n^2]$ 
and $w:=\sum_{i=1}^n |ii1\ket$. 
Identifying $\Tset$ with $[n]$, we can interpret a labeling 
$J\colon\Hy_n\to [n]$ with the filling of an $n\times n$ square with 
numbers in $[n]$. It is easy to see that 
$\eval_{\Hy_n}(J)=0$ unless $J$ is a 
{\em Latin square}, i.e., each number $j\in [n]$ occurs in 
each row and each column of the square exactly once. 
In this case, $J$ defines a permutation of $[n]$ in each 
row and each column of the Latin square. It is straightforward
to see that $\eval_{E^{(1)}}(J)$ equals the product of the 
signs of the row permutations and $\eval_{E^{(2)}}(J)$ 
equals the product of the signs of the column permutations.  
Moreover, $\eval_{E^{(3)}}(J)=1$. 
Let us call the Latin square {\em even} if 
$\eval_{E^{(1)}}(J) \cdot\eval_{E^{(2)}}(J) = 1$ and {\em odd}
if this value equals $-1$. 

Equation~\eqref{def:fR} implies that 
$f_{\Hy_n}(w)$ equals the difference of the number of 
even and the number of odd Latin squares. 

It is easy to see that $f_{\Hy_n}(w) = 0$ if $n$ is odd
(exchange two rows). 
The Alon-Tarsi Conjecture~\cite{AT:92} states 
$f_{\Hy_n}(w) \ne 0$ if $n$ is even.
For instance, this conjecture is known to be true 
for $n\le 24$ or if $n$ differs from an odd prime exactly by $1$, 
cf.~\cite{Dri:98,Gly:10}. The general case, however, is wide open.
We note that $f_{\Hy_n}\ne0$ iff $f_{\Hy_n}(w) \ne 0$. 
\end{example}

\begin{remark}\label{re:kumar}
The construction of explicit highest weight vectors in the 
polynomial scenario leads to questions regarding Latin squares 
and the Alon-Tarsi Conjecture as well, 
cf.\ Kumar~\cite{Kum:12}.  
\end{remark}

\begin{example}\label{ex:favorite}
The obstruction design $\Hy=[n]\times [n] \times [n]$ has 
the type $\la:=(n^2\times n,n^2\times n,n^2\times n)$.  
Since $N(\la)=1$,
Corollary~\ref{cor:kronest} implies 
$k(\la) \le 1$. 
Using known properties of Kronecker coefficients
(cf.\ \cite[\S4.5]{ike:12b}), we get 
$k(\la)= k(n^2\times n,n\times n^2,n\times n^2)$,
which equals the multiplicity of the $\GL_n\times\GL_n$\dash representation 
$\{n\times n^2\}\ot \{n\times n^2\}$
in the $\GL_{n^2}$\dash representation $\{n^2\times n\}$ upon 
restriction to $\GL_n\times\GL_n$. Since $\{n^2\times n\}$ 
stands for the $n$th power of the determinant, we get 
$k(\la)=1$. This implies $f_\Hy\ne 0$. 
(It is not obvious how to verify this directly.) 
Up to scaling, $f_\Hy\in \Sym^n\tensor^3 (\C^{n^2})^*$ 
is the unique polynomial satisfying the beautiful 
invariance property
$g f_\Hy = (\det g)^{-n} f_\Hy$,  
for $g\in\GL_{n^2}$. 
\end{example}

The following fundamental questions arise when studying 
the highest weight vectors~$f_\Hy$ labeled by obstruction designs~$\Hy$.

\begin{questions}
\label{que:obsdes}
 \begin{compactenum}[(1)]
  \item Given an obstruction design~$\Hy$ of type $\la \partintostar[n] d$ and a tensor $w \in \tensor^3\IZ^n$.
  What is the complexity of computing the evaluation $f_\Hy(w)$? 
  Is this problem $\sharpP$\dash hard under Turing reductions?
  \item Given an obstruction design~$\Hy$ of type $\la \partintostar[n] d$.
  What is the complexity of deciding whether $f_\Hy=0$?
  \item For a given partition triple $\la \partintostar[n] d$,
  explicitly describe a maximal linear independent subset of the set of obstruction designs of type~$\la$!
 \end{compactenum}
\end{questions}
Let $\Specht{\la^{(k)}}$ denote the irreducible $\aS_d$-representation corresponding to~$\la^{(k)}$ 
(Specht-module). 
An answer to Question~\ref{que:obsdes}(3) would result in an explicit basis of
$(\Specht{\la^{(1)}}\otimes\Specht{\la^{(2)}}\otimes\Specht{\la^{(3)}})^{\aS_d}$
and solve one of the most fundamental open questions in the representation theory of the symmetric groups.

\subsection{Determinantal Complexity}
We now turn from the tensor scenario to the polynomial scenario.
Our goal is to find polynomials in the vanishing ideal of~$\GL_{n^2}\det_n$
(compare~\cite{lamare:10} for an interesting result). 
For $\la \partinto[n^2]dn$, let $\pl \la d n$ denote the 
multiplicity of $\Weyl\la$ in the {\em plethysm} $\Sym^d\Sym^n\IC^{n^2}$.
From \cite[eq.~(5.2.6)]{BLMW:11} we know that
\begin{align}
\label{eq:multipl-det}
 \coord{\GL_{n^2}\det_n}_{\geq 0} = \bigoplus_{d \geq 0} \bigoplus_{\la \partinto[n^2] n d} \symmkron{\recta n d}{\la} \Weyl{\la^*},
\end{align}
where $\symmkron{\recta n d}{\la}$ is the \emph{symmetric Kronecker coefficients}, defined in \cite{BLMW:11}.
A sufficient criterion for the existence of a HWV of
weight $\la^*$ in the vanishing ideal $I(\GL_{n^2}\det_n)$ is given by
\begin{align}
\label{eq:plethsymmkron}
 \pl \la d n > \symmkron{\recta n d}{\la},
\end{align}
since 
$\mult_{\la^*}(I(\GL_{n^2}\det_n)) \ \stackrel{\eqref{eq:multipl-det}}{\geq} 
 \ \pl \la d n - \symmkron{\recta n d}{\la}$.

Here are two examples of partitions satisfying \eqref{eq:plethsymmkron}, 
found by computer calculations:
$(13,13,2,2,2,2,2)\partinto[7]36$ in degree $\tfrac {36}3 = 12$
and $(15,5,5,5,5,5,5)\partinto[7]45$ in degree $\tfrac{45}3 = 15$.
An abundance of other partitions satisfying \eqref{eq:plethsymmkron} is given in~\cite[Appendix]{ike:12b}.

The fact that a partition with 7 rows occurs in the vanishing ideal $I(\GL_9\det_3)_{12} \subseteq \Sym^{12}\Sym^3(\IC^9)^*$
implies that the same partition occurs in the intersection $I(\GL_9\det_3) \cap \Sym^{12}\Sym^3(\IC^7)^*$,
see the inheritance theorems in~\cite{BLMW:11}.
Hence we get $f \notin \overline{\GL_9\det_3}$
for Zariski almost all polynomials $f \in \Sym^3(\IC^7)^*$.
Note that an explicit construction and evaluation of the HWVs
in $\Sym^{d}\Sym^n\IC^\ell$ would directly give lower bounds on~$\docc$
for specific~$f$.

\section{Explicit HWVs}\label{se:explicit-HWVs}

The goal of this section is to prove Theorem~\ref{th:OD}.

\subsection{A Consequence of Schur-Weyl Duality}

The vector space $\tensor^d \IC^n$ is a $\GL_n \times \aS_d$\dash representation
via the commuting actions of $\aS_d$ and~$\GL_n$, defined for $\aS_d$ by 
\begin{align*}\label{eq:snaction}
 \pi (w_1 \otimes w_2 \otimes \cdots \otimes w_d) \coloneqq w_{\pi^{-1}(1)} \otimes \cdots \otimes w_{\pi^{-1}(d)} ,
  \quad \pi \in \aS_d ,
\end{align*}
and for $\GL_n$ as follows:
$$
 g(w_1 \ot w_2 \ot \cdots \ot w_d) \coloneqq  
gw_1 \otimes gw_2 \otimes \cdots \otimes gw_d, \quad g \in \GL_n.
$$
It follows that $\aS_d$ leaves the highest weight vector space $\HWV_\la(\tensor^d\IC^n)$
invariant. 

Recall that in~\eqref{eq:def-eval} we assigned to a set partition $\setp$ of $[d]$
a linear form $\eval_\setp$ on $\tensor^d \C^n$. 

\begin{proposition}\label{pro:generatePoly}
Let $\la\partinto[n]d$. 
If $\setp$ runs over all set partitions of $[d]$ 
with type $\la$, then the corresponding $\eval_\setp$ 
generate the vector space 
$\HWV_{\la^*}(\tensor^{d} (\C^n)^*)$
of highest weight vectors of weight $\la^*$. 
\end{proposition}

\begin{proof}
For $e=\{1,\ldots,\ell\}$, $\ell\in\N$,
the multilinear map $(\C^\ell)^\ell \to \C,\, J \mapsto\det|_e$
defines a linear form $\bra \widehat{\ell}|$ on $ \tensor^\ell\C^\ell$.  
It is obvious that $\bra\widehat\ell|$ is 
a HWV of weight $\recta \ell {(-1)}$.

Let $\mu:=\transpose \la$ denote 
the transposed partition of $\la$ and 
consider the following set partition of $[d]$ of type~$\la$: 
$$
 \setp_\la := \big\{ \{1,2,\ldots,\mu_1\}, \{\mu_1+1,\ldots,\mu_1+\mu_2\},\ldots \big\}.
$$
A moment's thought reveals that 
$\eval_{\setp_{\la}}  = \bigotimes_{i=1}^{\la_1} \bra \widehat{\mu_i}|$.
From this description, it is readily checked that $\eval_{\setp_{\la}}$ 
is a HWV of weight~$\la^*$.

All $\eval_{\setp}$ are obtained from from $\eval_{\setp_{\la}}$ 
by applying arbitrary permutations in $\aS_d$. 

Recall that $\Weyl{\la}$ and $\Specht\la$ denote the irreducible $\GL_n$\dash representation 
and irreducible $\aS_d$\dash representation corresponding to~$\la$, respectively.
The fundamental {\em Schur-Weyl duality} states that  
\[
 \tensor^d \IC^n \simeq \bigoplus_{\la \partinto[n] d} \Weyl{\la} \otimes \Specht{\la}
\]
as $\GL_n \times \aS_d$-representations, e.g., see~\cite[Sec.~4.2.4]{gw:09}.

Going over to the dual $W := (\IC^n)^*$ 
we obtain 
$\HWV_{\la^*}\left(\tensor^d W\right) \simeq \HWV_{\la^*}(\Weyl{\la^*}) \otimes \Specht{\la}^*$.
But $\HWV_{\la^*}(\Weyl{\la^*})$ is 1-dimensional (Lemma~\ref{lem:highestweightvector}) and
so, as $\aS_d$\dash representations, we have
$\HWV_{\la^*}\left(\tensor^d W\right) \simeq \Specht{\la}^*$, which is irreducible.
Hence the linear span of the $\aS_d$\dash orbit of $\eval_{\setp_{\la}}$ 
equals 
$\HWV_{\la^*}\left(\tensor^d W\right)$. 
\end{proof}

\subsection{Proof of Theorem~\ref{th:OD}}\label{se:pfthOD}

Obstruction designs can be looked at in different, equivalent ways.
Recall that an obstruction design $\Hy\subseteq [\ell_1]\times [\ell_2]\times [\ell_3]$ 
defines three set partitions $E^{(k)}$ of $\Hy$ satisfying the {\em intersection property} 
$|e^{(1)}\cap e^{(2)} \cap e^{(3)}| \le 1$ for all 
$(e^{(1)},e^{(2)},e^{(3)}) \in E^{(1)}\times  E^{(2)}\times  E^{(3)}$. 

Suppose now that $V$ is an abstract finite set endowed with three set partitions $\setp^{(k)}$ 
of the set~$V$ satisfying the above intersection property. Then the incidence structure 
$$
 \Hy:=\{ (e^{(1)},e^{(2)},e^{(3)}) \mid e^{(1)}\cap e^{(2)} \cap e^{(3)} \ne\emptyset\} 
  \subseteq \setp^{(1)}\times  \setp^{(2)}\times  \setp^{(3)}
$$
is an obstruction design (after numbering each of the sides $\setp^{(k)}$).  
This obstruction design allows to retrieve the set $V$ and the partitions $\setp^{(k)}$. 
In fact,  $\Hy\to V,\,(e^{(1)},e^{(2)},e^{(3)}) \mapsto v$ such that $\{v\}= e^{(1)}\cap e^{(2)} \cap e^{(3)}$ is a bijection.
Moreover, this maps the $1$-slice 
$\{(e^{(2)},e^{(3)}) \mid (e^{(1)},e^{(2)},e^{(3)}) \in \Hy\}$ to $e^{(1)}$. 
Similarly for the other slices. 

Now assume $V=[d]$, $\la \partintostar[n]d$, and suppose that 
$\setp^{(k)}$ is a set partition of $[d]$ of type $\la^{(k)}$, for $k=1,2,3$. 
Proposition~\ref{pro:generatePoly} implies that  
$\eval_{\setp^{(1)}} \ot \eval_{\setp^{(2)}} \ot \eval_{\setp^{(3)}}$ 
defines a highest weight vector of weight $\la^*$ in  
$\HWV_{\la^*}(\tensor^{d} \tensor^3(\C^n)^*)$. 
Moreover, these vectors span the highest weight vector space, when the 
$\setp^{(k)}$ independently run through all set partitions of $[d]$ of type $\la^{(k)}$. 

The linear forms on $\tensor^d\tensor^3 \C^n$ that are symmetric 
with respect to $\aS_d$ are obtained by composing the linear forms 
on $\tensor^d\tensor^3 \C^n$ with the symmetrization
\begin{equation}\label{eq:sym-map}
 \proj_d\colon \tensor^d \tensor^3 \C^n \twoheadrightarrow \Sym^d\tensor^3 \C^n 
\end{equation}
given by $\frac{1}{d!} \sum_{\pi \in \aS_d} \pi$. 
It follows that the linear forms 
$\big(\eval_{\setp^{(1)}} \ot \eval_{\setp^{(2)}} \ot \eval_{\setp^{(3)}}\big)\circ\proj_d $ 
generate the highest weight vector space 
$\HWV_{\la^*}(\Sym^{d} \tensor^3(\C^n)^*)$. 

If the three set partitions $\setp^{(k)}$ satisfy the above intersection property, 
then they define an obstruction design $\Hy$ by the above reasoning. 
Moreover, we have 
$f_\Hy = \big(\eval_{\setp^{(1)}} \ot \eval_{\setp^{(2)}} \ot \eval_{\setp^{(3)}}\big)\circ \proj_d $ 
by the definition of $f_\Hy$. 

To complete the proof of Theorem~\ref{th:OD}, it therefore suffices to show 
that if the intersection property is violated, then the resulting form vanishes. 

\begin{lemma}
\label{lem:zeropattern}
Suppose that there are $e^{(k)}\in\setp^{(k)}$ for $k=1,2,3$ such that 
$e^{(1)} \cap e^{(2)} \cap e^{(3)}$ contains more than one element. 
Then $\big(\eval_{\setp^{(1)}} \ot \eval_{\setp^{(2)}} \ot \eval_{\setp^{(3)}}\big)\circ \proj_d$ 
vanishes.

\end{lemma}

\begin{proof} 
Suppose that the distinct vertices $\yv$ and $\yv'$ are both contained in 
$e^{(1)} \cap e^{(2)} \cap e^{(3)}$.
Let $\tau\colon V(\Hy)\to V(\Hy)$ denote the transposition switching $\yv$ and~$\yv'$.
From a labeling $J^{(k)}\colon V(\Hy) \to \C^n$
we get a new labeling $J^{(k)} \circ \tau$ by composition of maps.
Recall that 
$\eval_{\Lambda^{(k)}}(J^{(k)}) = \prod_{e\in\Lambda^{(k)}} \det J^{(k)}|_e$. 
If $e\ne e^{(k)}$, then $\yv,\yv' \not\in e$ and 
$\det J^{(k)}|_e = \det (J^{(k)}\circ\tau)|_e$. 
On the other hand, if $e=e^{(k)}$, then $\yv,\yv'\in e$ and we obtain 
$\det J^{(k)}|_e = - \det (J^{(k)}\circ\tau)|_e$ 
since applying $\tau$ amounts to switching the columns 
indexed by $\yv$ and $\yv'$.
We conclude that 
$$
 \eval_{\setp^{(k)}}(J^{(k)})  = -\eval_{\setp^{(k)}}(J^{(k)}\circ \tau) .
$$
Writing
$F:=\eval_{\setp^{(1)}} \ot \eval_{\setp^{(2)}} \ot \eval_{\setp^{(3)}}$ 
we obtain
$F(J) = (-1)^3 F(J\circ\tau)$. 
It follows that 
$$
 \sum_{\pi\in\aS_d} F\big(\tensor^3 (J^{(k)}\circ\pi)\big) = 0 ,
$$ 
which completes the proof. 
\end{proof}

\section{Proof of Lemma 4.4}
\label{se:NVatMaMu}

Recall from Section~\ref{se:pfthOD} that we may interpret an obstruction design~$\Hy$ 
as a set $V(\Hy)$ endowed with three set partitions $E^{(k)}$ of~$V$ satisfying the intersection property. 

The obstruction design $\Hy_\kappa$ introduced in Section~\ref{ssec:lbMaMu} then can be 
visualized as follows (see Figure~\ref{fig:hookobstructiondesign}). 
The vertex set $V(\Hy)$ is partioned into disjoint sets $V^{(1)} \dotcup V^{(2)} \dotcup V^{(3)} \dotcup \{y^0\}$,
where $|V^{(k)}| = \kappa$ for all~$k$. 
Each $E^{(k)}$ consists of one hyperedge 
$e^{(k)} \coloneqq V^{(k+1)} \cup V^{(k+2)} \cup \{y^0\}$ of size $2\kappa+1$ (addition mod~$3$ in the exponent)
and $\kappa$ singletons. 
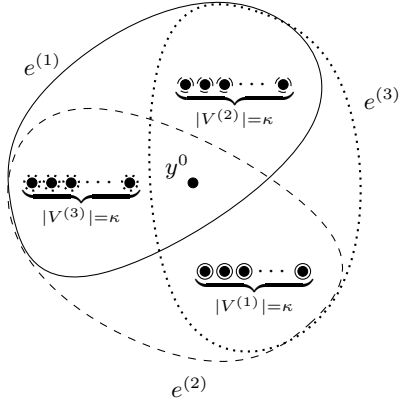
\begin{figure}[h]
\vspace{-0.35cm}
  \begin{center}
\begin{tikzpicture}[scale=0.8]
\draw[obsdesEdgeI] (-3.0,-0.5) .. controls (-2.86363,-1.13636) and (-2.59090,-1.77272) .. (-1.5,-1.5) .. controls (-0.40909,-1.22727) and (1.5,-0.04545) .. (2.0,1.0) .. controls (2.5,2.045454) and (1.590909,2.954545) .. (0.5,3.0) .. controls (-0.59090,3.045454) and (-1.86363,2.227272) .. (-2.5,1.5) .. controls (-3.13636,0.772727) and (-3.13636,0.136363) .. (-3.0,-0.5) -- cycle;
\draw[obsdesEdgeII] (1.933012,-2.34807) .. controls (2.415937,-1.91180) and (2.830681,-1.35742) .. (2.049038,-0.54903) .. controls (1.267394,0.259353) and (-0.71063,1.321765) .. (-1.86602,1.232050) .. controls (-3.02141,1.142336) and (-3.35416,-0.09950) .. (-2.84807,-1.06698) .. controls (-2.34198,-2.03446) and (-0.99705,-2.72759) .. (-0.04903,-2.91506) .. controls (0.898980,-3.10253) and (1.450087,-2.78435) .. (1.933012,-2.34807) -- cycle;
\draw[obsdesEdgeIII] (1.066987,2.848076) .. controls (0.447698,3.048163) and (-0.23977,3.130156) .. (-0.54903,2.049038) .. controls (-0.85830,0.967919) and (-0.78936,-1.27631) .. (-0.13397,-2.23205) .. controls (0.521415,-3.18779) and (1.763256,-2.85504) .. (2.348076,-1.93301) .. controls (2.932895,-1.01098) and (2.860692,0.500320) .. (2.549038,1.415063) .. controls (2.237383,2.329806) and (1.686276,2.647988) .. (1.066987,2.848076) -- cycle;
\node [obsdesnode,label={[label distance=-0.5em]above left:{\ensuremath{y^0}}}] at (0,0) {};

\node [label={[label distance=-0em]above:{\ensuremath{e^{(1)}}}}] at (-2.5,1.5) {};
\node [label={[label distance=-0em]below:{\ensuremath{e^{(2)}}}}] at (-0.04903,-2.91506) {};
\node [label={[label distance=-0em]right:{\ensuremath{e^{(3)}}}}] at (2.549038,1.415063) {};

\node [obsdesnode] at ($(-60:1.7cm)+(-2em,0)$) {};
\draw [obsdesEdgeI] ($(-60:1.7cm)+(-2em,0)$) circle (0.4em);
\node [obsdesnode] at ($(-60:1.7cm)+(-1em,0)$) {};
\draw [obsdesEdgeI] ($(-60:1.7cm)+(-1em,0)$) circle (0.4em);
\node [obsdesnode] at ($(-60:1.7cm)+(0em,0)$) {};
\draw [obsdesEdgeI] ($(-60:1.7cm)+(0em,0)$) circle (0.4em);
\node at ($(-60:1.7cm)+(1.5em,0)$) {$\cdots$};
\node [obsdesnode] at ($(-60:1.7cm)+(3em,0)$) {};
\draw [obsdesEdgeI] ($(-60:1.7cm)+(3em,0)$) circle (0.4em);
\node at ($(-60:1.7cm)+(0.5em,-1.3em)$) {$\underbrace{\hphantom{mmmmmi}}_{\tiny |V^{(1)}|=\kappa}$};

\node [obsdesnode] at ($(60:1.7cm)+(-3em,0.5em)$) {};
\draw [obsdesEdgeII] ($(60:1.7cm)+(-3em,0.5em)$) circle (0.4em);
\node [obsdesnode] at ($(60:1.7cm)+(-2em,0.5em)$) {};
\draw [obsdesEdgeII] ($(60:1.7cm)+(-2em,0.5em)$) circle (0.4em);
\node [obsdesnode] at ($(60:1.7cm)+(-1em,0.5em)$) {};
\draw [obsdesEdgeII] ($(60:1.7cm)+(-1em,0.5em)$) circle (0.4em);
\node at ($(60:1.7cm)+(0.5em,0.5em)$) {$\cdots$};
\node [obsdesnode] at ($(60:1.7cm)+(2em,0.5em)$) {};
\draw [obsdesEdgeII] ($(60:1.7cm)+(2em,0.5em)$) circle (0.4em);
\node at ($(60:1.7cm)+(-0.5em,-0.8em)$) {$\underbrace{\hphantom{mmmmmi}}_{\tiny |V^{(2)}|=\kappa}$};

\node [obsdesnode] at ($(180:1.7cm)+(-3em,0)$) {};
\draw [obsdesEdgeIII] ($(180:1.7cm)+(-3em,0)$) circle (0.4em);
\node [obsdesnode] at ($(180:1.7cm)+(-2em,0)$) {};
\draw [obsdesEdgeIII] ($(180:1.7cm)+(-2em,0)$) circle (0.4em);
\node [obsdesnode] at ($(180:1.7cm)+(-1em,0)$) {};
\draw [obsdesEdgeIII] ($(180:1.7cm)+(-1em,0)$) circle (0.4em);
\node at ($(180:1.7cm)+(0.5em,0)$) {$\cdots$};
\node [obsdesnode] at ($(180:1.7cm)+(2em,0)$) {};
\draw [obsdesEdgeIII] ($(180:1.7cm)+(2em,0)$) circle (0.4em);
\node at ($(180:1.7cm)+(-0.5em,-1.3em)$) {$\underbrace{\hphantom{mmmmmi}}_{\tiny |V^{(3)}|=\kappa}$};
\end{tikzpicture}
\vspace{-0.3cm}
    \caption{The unique family of obstruction designs corresponding to the hook partition triple $\la(\kappa)$.} 
    \nopar\label{fig:hookobstructiondesign}
  \end{center}
\vspace{-0.5cm}
\end{figure}


We outline now the proof of Lemma~\ref{le:NEU}. 
For notational convenience, we define the triples of vectors 
\begin{align}
\label{eq:mamutriples}
\tri{ij\elll} \coloneqq \big(|ij\ket,|j\elll\ket,|\elll i\ket\big) \in (\IC^{m \times m})^3
\end{align}
(omitting parentheses) and  put $\Tset \coloneqq \{ \tri{ij\elll} \mid 1 \leq i,j,\elll \leq m \}$.
Recall from \eqref{eq:defmamu}:
\[
 \MaMu_m = \sum_{i,j,\elll=1}^m \tri{ij\elll}^{(1)} \otimes \tri{ij\elll}^{(2)} \otimes \tri{ij\elll}^{(3)}.
\]
Let $A^{(k)}\colon \C^{m\times m}\to \C^{m\times m}$ be linear maps and 
$A=(A^{(1)},A^{(2)},A^{(3)})$.  
For a triple labeling 
$J\colon V(\Hy)\to\Tset$ we define the composed triple labeling 
$AJ\colon V(\Hy)\to (\C^{m\times m})^3$ by 
$(AJ)^{(k)}(y) := A^{(k)}(J^{(k)}(y) )$ for $y\in V(\Hy)$.  

After fixing a numbering of the vertices of $\Hy$, 
Equation~\eqref{def:fR} can be written as 
\begin{align}\textstyle
\label{eq:evalatmamu}
\tag{\ensuremath{\dagger}}
f_\Hy(A\MaMu_m)
= \sum_{J\colon V(\Hy) \to \Tset} \eval_\Hy(AJ) .
\end{align}
The strategy is to construct a triple $A$
of $m^2 \times m^2$ matrices
having affine linear entries in the indeterminates 
$\omegaX_1,\ldots,X_N$ 
with the property that the coefficient of a specific monomial~$\omegaXM$ in the $X_i$
in $f_\Hy(A\MaMu_m)$ is nonzero.
Hence $f(X_1,\ldots,X_N):=f_\Hy(A\MaMu_m)$ is not the zero polynomial in the $X_i$. 
By perturbing the $A^{(k)}$ we may assume w.l.o.g.\ that all 
$A^{(k)}$ are invertible. 
There is a substitution of the~$X_i$ with suitable values 
$\alpha_1,\ldots,\alpha_N\in\C$
such that $f(\alpha)\ne 0$. Making this substitution in $A$ yields
the desired matrix triple over $\C$.


\subsection{Invariance in each $V^{(k)}$}
We use the short notation 
$\eval_{e}(J) \coloneqq \det J^{(k)}|_e$  
for a hyperedge $e \in E^{(k)}$ and a triple labeling~$J$.

\begin{claim}
\label{cla:invarineachVk}
 Let $\sigma\colon V(\Hy)\to V(\Hy)$ be a bijection satisfying
 $\sigma(V^{(k)})=V^{(k)}$ for all $k\in\{1,2,3\}$.
 For every triple labeling $J\colon V(\Hy) \to (\IC^{m^2})^3$ we have
 $\eval_\Hy(J) = \eval_\Hy(J \circ \sigma)$.
\end{claim}

\begin{proof}
It suffices to show the claim for a transposition~$\sigma$ exchanging two elements of $V^{(1)}$,
because the situation for $V^{(2)}$ and $V^{(3)}$ is completely symmetric.
We have $\prod_{e \in E^{(1)}} \eval_e(J) = \prod_{e \in E^{(1)}} \eval_e(J\circ\sigma)$,
because, up to reordering, both products have the same factors.
For $k\in\{2,3\}$ we have
$\eval_{e}(J) = \eval_{e}(J\circ\sigma)$ for every singleton hyperedge~$e \in E^{(k)}$
and $\eval_{e^{(k)}}(J) = -\eval_{e^{(k)}}(J\circ\sigma)$.
Therefore
$\prod_{e \in E^{(k)}} \eval_e(J) = -\prod_{e \in E^{(k)}} \eval_e(J\circ\sigma)$.
As a result we get
$\eval_\Hy(J) = (-1)^2\eval_\Hy(J\circ\sigma).$
\end{proof}

\subsection{Special Structure of the Matrix Triple}

Let $\Gamma \coloneqq \IC[\omegaX_{i}^{(k)} : 1 \leq k \leq 3, \, 1 \leq i \leq m ]$
denote the polynomial ring in $3m$ variables.
Recall that $m$ is odd and $\kappa = \frac {m^2-1}2$.
We set $\bar i \coloneqq m+1-i$ for $1 \leq i \leq m$, thinking of 
$i \mapsto \bar i\vphantom{\big(\big)}$ as a reflection at  $a \coloneqq (m+1)/2$. 
Note $\bar a = a$. 
We consider the set of pairs $\Tup \coloneqq \{1,\ldots,m\} \times \{1,\ldots,m\} \setminus \{aa\}$
and fix an arbitrary bijection $\varphi\colon \Tup \to \{2,\ldots,m^2\}$.

For each $1 \leq k \leq 3$ we define the 
matrix $A^{(k)}$ of format $(m\times m)\times m^2$
with the following affine linear entries in $X_i^{(k)}$:
\[
A^{(k)}|ij\ket \coloneqq
\begin{cases}
 \omegaX^{(k)}_{a} |1\ket& \text{ if } i=j=a\\
 |\varphi(i\bar i)\ket + \omegaX_{i}^{(k)}|1\ket & \text{ if } i \neq j \text{ and } j = \bar i \\
 |\varphi(ij)\ket & \text{ if } j \neq \bar i
\end{cases}.\]
Hence $A^{(k)}$ looks as follows:
\begin{align}
 \tag{\textreferencemark}
 \label{eq:matrixtriple}
\left(\begin{array}{ccccccc|c}
\!\!\omegaX_a^{(k)}\!\! & \!\!\omegaX_1^{(k)}\!\!   & \!\!\cdots\!\! & \!\!\omegaX_{a-1}^{(k)}\!\! & \!\!\omegaX_{a+1}^{(k)}\!\!& \!\!\cdots\!\! & \!\!\omegaX_{m}^{(k)}\!\! \\
 &1& & & && & \\
 & & \!\!\ddots\!\!&&  & &  & \\
 & &  &1&&   &  &\smash{\raisebox{0.3cm}{\text{\LARGE $0$}}}\\
 & &  & &1&  &  &  \\
 & & & & &  \!\!\ddots\!\!  &  & \\
 & & & && &1 &  \\
\hline
 & &    & &&& &\\
 & & &\smash{\text{\LARGE $0$}} &  & &  & \id_{m^2-m}\\
 & &  && & &  &
\end{array}\right),
\end{align}
where we arranged the rows and columns as follows:
The left $m$ columns correspond to the vectors~$|i\bar i\ket$,
where the leftmost one corresponds to~$|aa\ket$.
The top row corresponds to the vector~$|1\ket$
and the following $m-1$ rows correspond to the vectors $|\varphi(i\bar i)\ket$.
Recall that $f_\Hy(A\MaMu_m)$
is a sum of products of determinants of submatrices of the $A^{(k)}$.

The sum $f_\Hy(A\MaMu_m)$ is an element of $\Gamma$
and we are interested in its coefficient of the monomial $\omegaXM$, where
\begin{align}
\label{eq:omegamonomial}
 \omegaXM \coloneqq \prod_{k=1}^3 \omegaX^{(k)}_{a} \prod_{i=1}^m \big(\omegaX_{i}^{(k)}\big)^{|i - \bar i|}.
\end{align}
We remark that the degree of $\omegaXM$ is $3( 1+\sum_{i=1}^m |i-\bar i|)$.
It is readily checked that $\sum_{i=1}^m |i-\bar i| = \kappa$.

We call a triple labeling $J\colon V(\Hy) \to \Tset$
\emph{nonzero},
if the coefficient of $\omegaXM$ in the polynomial
$\eval_\Hy(AJ)$ is nonzero.
We will count and classify all nonzero triple labelings~$J$ and show that 
all $\eval_\Hy(AJ)$ contribute the same coefficient with respect to the 
monomial~$\omegaXM$.
This implies that the coefficient of~$\omegaXM$ in $f_\Hy(A\MaMu_m)$
is a sum without cancellations and hence is nonzero.

\subsection{Separate Analysis of the Three Layers}
We fix a nonzero triple labeling $J\colon V(\Hy) \to \Tset$ 
and write $J=(J^{(1)},J^{(2)},J^{(3)})$.
Recall that the hyperedge $e^{(k)}$ has size $2\kappa+1 = m^2$.
Since $J$ is nonzero, 
$J^{(k)}$ is injective on hyperedges and therefore
$|\{J^{(k)}(\yv) : \yv \in e^{(k)}\}| = m^2$.
Hence $J^{(k)}$ is bijective on $e^{(k)}$.

\begin{claim}
\label{cla:labelinV}
For all $\yv \in V^{(k)}$ we have
$J^{(k)}(\yv) = |i \bar i\ket$ for some $1 \leq i \leq m$.
\end{claim}

\begin{proof}
Since $\{y\} \in E^{(k)}$ and $J$ is nonzero, 
we have $\bra 1 | A^{(k)} |J^{(k)}(y)\ket \neq 0$.
From the definition of~$A$ it follows that $J^{(k)}(y) = |ij\ket$ and the third case $j \neq \bar i$
is excluded. Hence $j = \bar i$.
\end{proof}

\begin{claim}
\label{cla:centerlabel}
We have
$J(\yv^0) = (|aa\ket,|aa\ket,|aa\ket)$.
\end{claim}

\begin{proof}
For the following argument it is important to keep the structure of the matrix $A^{(k)}$ in mind,
cf.~\eqref{eq:matrixtriple}.
Recall that $f_\Hy(A \MaMu_m)$ is a sum of products of certain subdeterminants of~$A^{(k)}$
that are determined by the hyperedges in $E^{(k)}(\Hy)$.
The coefficient of~$\omegaXM$ in $\eval_\Hy(AJ(1),\ldots,AJ(d))$ is nonzero
as $J$ is nonzero. Fix~$k$. 
Since the degree of $\omegaX_a^{(k)}$ in~$\omegaXM$ is one, 
there is exactly one vertex $\yv_k \in V(\Hy)$ with $J^{(k)}(y) = |aa\ket$.
But we know that $J^{(k)}$ bijective on $e^{(k)}$, so $\yv_k \in e^{(k)}$.

It is now sufficient to show that $\yv_1=\yv_2=\yv_3$
(since $e^{(1)}\cap e^{(2)}\cap e^{(3)} = \{\yv^0\}$). 


The structure of the matrix multiplication tensor implies
that $J(\yv_1) = (|aa\ket,|ai\ket,|ia\ket)$ for some $1 \leq i \leq m$.

In the case $a=i$, by definition of $\yv_2$ and $\yv_3$ and uniqueness, we have
$\yv_1=\yv_2=\yv_3$ and we are done.

So consider the case where $a \neq i$.
If $\yv_1 \neq y^0$ we may assume 
w.l.o.g.\ $\yv_1 \in V^{(3)}$.
Using Claim~\ref{cla:labelinV} we conclude that
$J^{(3)}(\yv_1) = |i\bar i\ket$ for some $1 \leq i \leq m$.
Hence $\bar i = a$ contradicting $i \neq a$.
So we must have $\yv_1 = \yv^0$.

Similarly, we show that $\yv_2 = \yv_3 = \yv^0$
and the assertion follows. 
\end{proof}

\begin{claim}
\label{cla:mappinglabel}
We have
$J^{(k)}(V^{(k)}) = \{ |i\bar i\ket \mid 1 \leq i \leq m \}\setminus\{|aa\ket\}$,
 where the preimage of each $|i\bar i\ket$ under~$J^{(k)}$ has size $|i-\bar i|$.
\end{claim}

\begin{proof}
According to Claim~\ref{cla:centerlabel} we have 
$J(\yv^0) = (|aa\ket,|aa\ket,|aa\ket)$.
Since $A^{(k)}|aa\ket$ is a multiple of $|1\ket$,
$\eval_{e^{(k)}}(J^{(k)})$ is a multiple of~$\omegaX_a^{(k)}$, 
cf.~\eqref{eq:matrixtriple}.
Moreover, for $i \neq a$, the variable $\omegaX_i^{(k)}$ does not appear in
the expansion of~$\eval_{e^{(k)}}(J^{(k)})$.
Since there are $\kappa = \sum_{i=1}^m|i-\bar i|$ many
contributions of a factor $\omegaX_i^{(k)}$ in the monomial $\omegaXM$,
these factors must be contributed at vertices in~$V^{(k)}$.
Moreover $|V^{(k)}|=\kappa$, so the only possibility is that all $\yv \in V^{(k)}$ satisfy
$J^{(k)}(\yv) = |i \bar i\ket$ for some $1 \leq i \leq m$, $i \neq a$.
The specific requirement for the number of factors $\omegaX_i^{(k)}$
which are encoded in~$\omegaXM$ in~\eqref{eq:omegamonomial} finishes the proof.
\end{proof}

\subsection{Coupling the Analysis of the Three Layers}

It will be convenient to identify the sets $J^{(k)}(V^{(k')})$ 
with their corresponding subsets of $\Tup$. 

Consider the bijective map
$
  \tau\colon \Tup \to \Tup, \ \tau(ij) = (j\bar i),
$
which corresponds to the rotation by~$90^\circ$.
Clearly, $\tau^4 = \id$.
The map $\tau$ induces a map $\wp(\Tup)\to\wp(\Tup)$ on the powerset,
which we also denote by~$\tau$.

Taking the complement defines the involution 
$
  \iota\colon \wp(\Tup) \to \wp(\Tup), \ S \mapsto \Tup \setminus S.
$
Clearly, we have $\tau \circ \iota = \iota \circ \tau$.
We will only be interested in subsets $S \subseteq \Tup$ with exactly $|\Tup|/2 = \kappa$ many elements
and their images under $\tau$ and~$\iota$.
The subsets $S\subseteq \Tup$ that satisfy $\iota(S) = \tau(S)$ will be of special interest.
Geometrically, these are the sets that get inverted when rotating by~$90^\circ$.

In Claim~\ref{cla:mappinglabel} we analyzed the labels $J^{(k)}(V^{k})$.
In the next claim we turn to $J^{(k)}(V^{k'})$, where $k \neq k'$.

\begin{claim}
\label{cla:tauiota}
Every nonzero triple labeling~$J$ is completely determined by the image $J^{(1)}(V^{(3)})$
(up to permutations in the $V^{(k)}$, see Claim~\ref{cla:invarineachVk}) as follows.
\begin{compactitem}
 \item $J^{(2)}(V^{(3)}) = \tau(J^{(1)}(V^{(3)}))$,
 \item $J^{(2)}(V^{(1)}) = \iota(J^{(2)}(V^{(3)}))$,
 \item $J^{(3)}(V^{(1)}) = \tau(J^{(2)}(V^{(1)}))$,
 \item $J^{(3)}(V^{(2)}) = \iota(J^{(3)}(V^{(1)}))$,
 \item $J^{(1)}(V^{(2)}) = \tau(J^{(3)}(V^{(2)}))$.
\end{compactitem}
Moreover, $\tau(J^{(1)}(V^{(3)})) = \iota(J^{(1)}(V^{(3)}))$.
\end{claim}

\begin{proof}
According to Claim~\ref{cla:mappinglabel}, 
each vertex $\yv \in V^{(3)}$ satisfies
\begin{align*}
J(\yv) = \big(|ij\ket,|\tau (ij)\ket,|\bar i i\ket\big) 
\end{align*}
for some $1 \leq i,j \leq m$, $i \neq a$.
In particular, 
\begin{align*}
 \tau(J^{(1)}(V^{(3)})) = J^{(2)}(V^{(3)}).
\end{align*}
Recall that 
$J^{(2)}$ is bijective on~$e^{(2)}$.
Using $e^{(2)} = V^{(1)} \dotcup V^{(3)} \dotcup \{\yv^0\}$ we see that
\[
J^{(2)}(V^{(1)}) = \Tup \setminus J^{(2)}(V^{(3)}) = \iota(J^{(2)}(V^{(3)})).
\]
For the same reason, we can deduce
$J^{(3)}(V^{(1)}) = \tau(J^{(2)}(V^{(1)}))$ and
$J^{(3)}(V^{(2)}) = \iota(J^{(3)}(V^{(1)}))$.
And applying these arguments one more time we get
$J^{(1)}(V^{(2)}) = \tau(J^{(3)}(V^{(2)}))$
and
$J^{(1)}(V^{(3)}) = \tau(J^{(1)}(V^{(2)}))$.
Summarizing (recall $\tau \circ \iota = \iota \circ \tau$) we have
\[J^{(1)}(V^{(3)}) = \tau^3\iota^3(J^{(1)}(V^{(3)})) = \tau^{-1}\iota(J^{(1)}(V^{(3)})),\]
which is equivalent to $\tau(J^{(1)}(V^{(3)})) = \iota(J^{(1)}(V^{(3)}))$.
\end{proof}

%
\begin{definition}
\label{def:validset}
 A subset $\Tups \subseteq \Tup$ is called \emph{valid},
 if
 \begin{enumerate}[(1)]
  \item $|\Tups|=\frac{m^2-1}2 = \kappa$,\label{def:validset:I}
  \item $\tau(\Tups) = \iota(\Tups)$,\label{def:validset:II}
  \item $|p^{-1}(i)| = |i-\bar i|$ for all $1 \leq i \leq m$\label{def:validset:III}
 \end{enumerate}
where $p\colon\Tups \to \{1,\ldots,m\}$ is the projection to the first component.
\end{definition}

\begin{proposition}
\label{pro:nonzerovalid}
$J^{(1)}(V^{(3)})$ is a valid set for all nonzero triple labelings~$J$.
On the other hand, for every valid set~$\Tups$ there exists exactly one nonzero triple labeling~$J$
 with $J^{(1)}(V^{(3)}) = \Tups$, up to permutations in the~$V^{(k)}$.
\end{proposition}

\begin{proof}
 For the first statement, property (\ref{def:validset:II}) of Def.~\ref{def:validset}
 follows from Claim~\ref{cla:tauiota} and
 property (\ref{def:validset:III}) of Def.~\ref{def:validset} follows from Claim~\ref{cla:mappinglabel}.
 The second statement can be readily checked with Claim~\ref{cla:centerlabel} and Claim~\ref{cla:tauiota}.
\end{proof}
Figure~\fref{fig:combinatorialsets} gives an example for the case $m=9$.
    Vertices that appear in all valid sets are drawn with a solid border.
    Vertices that appear in no valid set are drawn with a dotted border.
    Vertices that appear in half of all valid sets are drawn with a dashed border.
    These contain a vertex label $x_i$ or $\overline{x_i}$.
    Each valid set corresponds to a choice vector $x \in \{\text{true},\text{false}\}^4$
    determining whether the $x_i$ or the $\overline{x_i}$ are contained in~$\Tups$.
    This results in $2^4=16$ valid sets $\Tups \subseteq \Tup$.

\begin{figure}[h]
  \begin{center}
 \begin{tikzpicture}[scale=0.9, x={(0cm,-0.65cm)},y={(0.65cm,0cm)},verte/.style={%
        circle,
        draw=black,
        inner ysep=2pt,
        inner xsep=2pt,
        minimum width = 0.5cm,
        minimum height = 0.5cm}]
        \node at (1,-0.5) {\ensuremath{i}};
        \node at (-0.5,1) {\ensuremath{j}};
        \draw[-latex] (0,0) -- (0,1);
        \draw[-latex] (0,0) -- (1,0);
        \node[verte,dashed] at (1,1) {\tiny$\truthlabel_1$};
        \node[verte] at (1,2) {};
        \node[verte] at (1,3) {};
        \node[verte] at (1,4) {};
        \node[verte] at (1,5) {};
        \node[verte] at (1,6) {};
        \node[verte] at (1,7) {};
        \node[verte] at (1,8) {};
        \node[verte,dashed] at (1,9) {\tiny$\overline{\truthlabel_1}$};
        \node[verte,dotted] at (2,1) {};
        \node[verte,dashed] at (2,2) {\tiny$\truthlabel_2$};
        \node[verte] at (2,3) {};
        \node[verte] at (2,4) {};
        \node[verte] at (2,5) {};
        \node[verte] at (2,6) {};
        \node[verte] at (2,7) {};
        \node[verte,dashed] at (2,8) {\tiny$\overline{\truthlabel_2}$};
        \node[verte,dotted] at (2,9) {};
        \node[verte,dotted] at (3,1) {};
        \node[verte,dotted] at (3,2) {};
        \node[verte,dashed] at (3,3) {\tiny$\truthlabel_3$};
        \node[verte] at (3,4) {};
        \node[verte] at (3,5) {};
        \node[verte] at (3,6) {};
        \node[verte,dashed] at (3,7) {\tiny$\overline{\truthlabel_3}$};
        \node[verte,dotted] at (3,8) {};
        \node[verte,dotted] at (3,9) {};
        \node[verte,dotted] at (4,1) {};
        \node[verte,dotted] at (4,2) {};
        \node[verte,dotted] at (4,3) {};
        \node[verte,dashed] at (4,4) {\tiny$\truthlabel_4$};
        \node[verte] at (4,5) {};
        \node[verte,dashed] at (4,6) {\tiny$\overline{\truthlabel_4}$};
        \node[verte,dotted] at (4,7) {};
        \node[verte,dotted] at (4,8) {};
        \node[verte,dotted] at (4,9) {};
        \node[verte,dotted] at (5,1) {};
        \node[verte,dotted] at (5,2) {};
        \node[verte,dotted] at (5,3) {};
        \node[verte,dotted] at (5,4) {};
        \node[verte,dotted] at (5,6) {};
        \node[verte,dotted] at (5,7) {};
        \node[verte,dotted] at (5,8) {};
        \node[verte,dotted] at (5,9) {};
        \node[verte,dotted] at (6,1) {};
        \node[verte,dotted] at (6,2) {};
        \node[verte,dotted] at (6,3) {};
        \node[verte,dashed] at (6,4) {\tiny$\overline{\truthlabel_4}$};
        \node[verte] at (6,5) {};
        \node[verte,dashed] at (6,6) {\tiny$\truthlabel_4$};
        \node[verte,dotted] at (6,7) {};
        \node[verte,dotted] at (6,8) {};
        \node[verte,dotted] at (6,9) {};
        \node[verte,dotted] at (7,1) {};
        \node[verte,dotted] at (7,2) {};
        \node[verte,dashed] at (7,3) {\tiny$\overline{\truthlabel_3}$};
        \node[verte] at (7,4) {};
        \node[verte] at (7,5) {};
        \node[verte] at (7,6) {};
        \node[verte,dashed] at (7,7) {\tiny$\truthlabel_3$};
        \node[verte,dotted] at (7,8) {};
        \node[verte,dotted] at (7,9) {};
        \node[verte,dotted] at (8,1) {};
        \node[verte,dashed] at (8,2) {\tiny$\overline{\truthlabel_2}$};
        \node[verte] at (8,3) {};
        \node[verte] at (8,4) {};
        \node[verte] at (8,5) {};
        \node[verte] at (8,6) {};
        \node[verte] at (8,7) {};
        \node[verte,dashed] at (8,8) {\tiny$\truthlabel_2$};
        \node[verte,dotted] at (8,9) {};
        \node[verte,dashed] at (9,1) {\tiny$\overline{\truthlabel_1}$};
        \node[verte] at (9,2) {};
        \node[verte] at (9,3) {};
        \node[verte] at (9,4) {};
        \node[verte] at (9,5) {};
        \node[verte] at (9,6) {};
        \node[verte] at (9,7) {};
        \node[verte] at (9,8) {};
        \node[verte,dashed] at (9,9) {\tiny$\truthlabel_1$};
  \end{tikzpicture}
  \vspace{-0.2cm}
    \caption{The case $n=9$.}
    \nopar\label{fig:combinatorialsets}
  \end{center}
  \vspace{-0.4cm}
\end{figure}
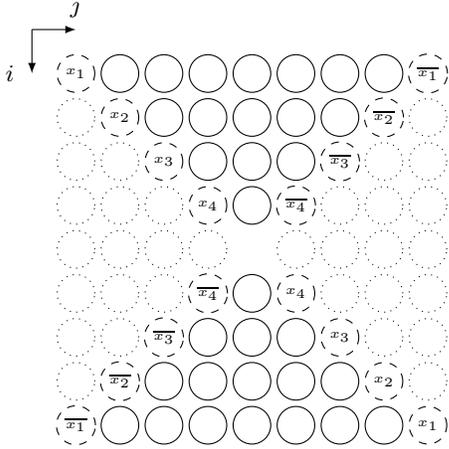

The next claim classifies all valid sets.
\begin{lemma}
\label{lem:classifyvalid}
A set $\Tups \subseteq \Tup$ is valid iff the following conditions are all satisfied (see Figure~\ref{fig:combinatorialsets} for an illustration):
 \begin{enumerate}[(1)]
  \item $\Big\{(ij) \mid (i<j \textup{ and } i<\bar j) \textup{ or } (i>j \textup{ and } i>\bar j)\Big\} \subseteq \Tups$,
  represented by solid vertices in Figure~\ref{fig:combinatorialsets}.
  \item $\Big\{(ij) \mid (i>j \textup{ and } i<\bar j) \textup{ or } (i<j \textup{ and } i>\bar j)\Big\} \cap S = \emptyset$,
  represented by dotted vertices in Figure~\ref{fig:combinatorialsets}.
  \item For all $1 \leq i \leq \frac{m-1}2$ there are two mutually exclusive cases, (a) and (b),
  represented by the two vertices $\truthlabel_i$ and the two vertices $\overline{\truthlabel_i}$,
  respectively, in Figure~\ref{fig:combinatorialsets}.
  \begin{samepage}
  \begin{enumerate}[(a)]
   \item $\{(ii),(\bar i\bar i)\} \subseteq \Tups$ \ and \ $\{(i\bar i),(\bar i i)\}\cap\Tups = \emptyset$,
   \item $\{(i\bar i),(\bar i i)\} \subseteq \Tups$ \ and \ $\{(ii),(\bar i\bar i)\}\cap\Tups = \emptyset$.
  \end{enumerate}
  These choices result in $2^{\frac{m-1}{2}}$ valid sets.
  \end{samepage}
 \end{enumerate}
\end{lemma}

\begin{proof}
As indicated in Figure~\ref{fig:combinatorialsets},
for each tuple $(ij)$ we call $i$ the \emph{row} of $(ij)$.
For $\Tups$ to be valid,
according to Def.~\ref{def:validset}(\ref{def:validset:III}),
$\Tups$ must contain $|i-\bar i|$ elements in row $i$
and
according to Def.~\ref{def:validset}(\ref{def:validset:II}),
$\tau(\tup) \notin \Tups$ for all $\tup \in \Tups$.

In particular, $\Tups$ must contain $m-1$ elements in row 1.
If $(11) \in \Tups$, then $(1m) \notin \Tups$,
because $\tau(11)=(1m)$.
Hence there are only two possibilities:
\textbf{(a):} $\{ (1j) \mid 1 \leq j < m \} \subseteq \Tups$
or
\textbf{(b):}
$\{ (1j) \mid 1 < j \leq m \} \subseteq \Tups$.
By symmetry,
for row~$m$ we get
\textbf{(a'):}
$\{ (mj) \mid 1 \leq j < m \} \subseteq \Tups$
or
\textbf{(b'):}
$\{ (mj) \mid 1 < j \leq m \} \subseteq \Tups$.
But since $\tau(1m) = (mm)$ and $\tau(m1) = (11)$,
the fact $\tau(\Tups) = \iota(\Tups)$
implies that
\textbf{(a)} iff \textbf{(b')}
and that \textbf{(a')} iff \textbf{(b)}.
We are left with the two possibilities
$\big($\textbf{(a)} and \textbf{(b')}$\big)$
or $\big($\textbf{(a')} and \textbf{(b)}$\big)$.

Now consider row 2. We have $\tau(21) = (1,m-1) \in \Tups$
and hence $(21) \notin \Tups$.
In the same manner we see $(2m) \notin \Tups$.
We are left to choose $m-3$ elements from the $m-2$ remaining elements in row 2.
The same argument as for row 1 gives two possibilities:
\textbf{(a):} $\{ (2j) \mid 2 \leq j < m-1 \} \subseteq \Tups$
or
\textbf{(b):}
$\{ (2j) \mid 2 < j \leq m-1 \} \subseteq \Tups$.
Analogously for row $m-1$ we have
\textbf{(a'):} $\{ (m-1,j) \mid 2 \leq j < m-1 \} \subseteq \Tups$
or
\textbf{(b'):}
$\{ (m-1,j) \mid 2 < j \leq m-1 \} \subseteq \Tups$.
With the same reasoning
as for the rows $1$ and $m$ we get
\textbf{(a)} iff \textbf{(b')}
and that \textbf{(a')} iff \textbf{(b)}.
Again we are left with the two possibilities
$\big($\textbf{(a)} and \textbf{(b')}$\big)$
or $\big($\textbf{(a')} and \textbf{(b)}$\big)$.

Continuing these arguments we end up with $2^{\frac{m-1}2}$ possibilities.
It is easy to see that each of these possibilities gives a valid set.
\end{proof}

The following claim finishes the proof of 
Lemma~\ref{le:NEU}.

\begin{claim}
All nonzero triple labelings~$J$ have the same coefficient of $\omegaXM$
in $\eval_\Hy(AJ)$.
\end{claim}

\begin{proof}
Take two nonzero triple labelings $J$ and~$J'$.
According to Proposition~\ref{pro:nonzerovalid},
both sets $J^{(1)}(V^{(3)})$ and ${J'}^{(1)}(V^{(3)})$ are valid sets.
Because of Lemma~\ref{lem:classifyvalid},
it suffices to consider only the case where $J^{(1)}(V^{(3)})$ and $J'^{(1)}(V^{(3)})$
differ by a single involution $\sigma\colon \Tup \to \Tup$,
where for some fixed $1 \leq i \leq \frac{m-1}2$ we have
$\sigma(ii)=(i \bar i)$ and $\sigma(\bar i\bar i)=(\bar i i)$,
and $\sigma$ is constant on all other pairs.

We analyze the labels that are affected by~$\sigma$.
We only perform the analysis for one of the two symmetric cases,
namely for $\{|ii\ket, |\bar i\bar i\ket\} \subseteq J^{(1)}(V^{(3)})$.
Note that this implies
\begin{align}
\label{eq:ibariformat}
\tag{\ensuremath{\diamondsuit}}
\big\{\big(|ii\ket,|i\bar i\ket,|\bar i i\ket\big), 
\big(|\bar i\bar i\ket,|\bar i i\ket,|i \bar i\ket\big)\big\} \subseteq J(V^{(3)}),
\end{align}
according to Claim~\ref{cla:mappinglabel}.
We adapt the notation from \eqref{eq:mamutriples} to our special situation
and write $\tri{000} \coloneqq \tri{\bar i \bar i\bar i}$, \
$\tri{001} \coloneqq \tri{\bar i\bar i i}$, \ $\ldots$, \ $\tri{111} \coloneqq \tri{iii}$.
Using this notation, \eqref{eq:ibariformat}~reads as follows:
$\{\tri{110},\tri{001}\}\subseteq J(V^{(3)})$. Using Claim~\ref{cla:tauiota}
we get
\[
\{\tri{101},\tri{010}\}\subseteq J(V^{(2)}), \ \
\{\tri{011},\tri{100}\}\subseteq J(V^{(1)}).
\]
Applying $\sigma$ to $J^{(1)}(V^{(3)})$,
we can use Claim~\ref{cla:mappinglabel} again to get
\[
\big\{\big(|i\bar i\ket,|\bar i\bar i\ket,|\bar i i\ket\big),\big(|\bar i i\ket,|i i\ket,|i \bar i\ket\big) \big\}\subseteq J'(V^{(3)}).
\]
Applying Claim~\ref{cla:tauiota} and using our short syntax, we get:
\begin{align*}
 \{\tri{100},\tri{011}\}\subseteq J'(V^{(3)}), \\
 \{\tri{001},\tri{110}\}\subseteq J'(V^{(2)}), \\
 \{\tri{010},\tri{101}\}\subseteq J'(V^{(1)}).
\end{align*}
We see that exactly the same triples occur in~$J(V(\Hy))$ as in $J'(V(\Hy))$.
We focus now on $J^{(1)}$ and $J'^{(1)}$ and see that:
\[
 \{(ii),(\bar i\bar i) \}\subseteq J^{(1)}(V^{(3)}) \text{ and } \{(i\bar i),(\bar i i) \}\subseteq J^{(1)}(V^{(2)})
\]
and
\[
 \{(i\bar i),(\bar i i) \}\subseteq {J'}^{(1)}(V^{(3)}) \text{ and } \{(\bar i \bar i),(ii) \}\subseteq {J'}^{(1)}(V^{(2)}).
\]
This gives exactly two switches of positions in $e^{(1)} = V^{(2)} \dotcup V^{(3)} \dotcup \{\yv^0\}$,
hence \[\eval_{e^{(1)}}(AJ) = (-1)^2\eval_{e^{(1)}}(AJ') = \eval_{e^{(1)}}(AJ').\]
Analogously we can prove that $\eval_{e^{(k)}}(AJ) = \eval_{e^{(k)}}(AJ')$
for all $k\in\{2,3\}$ 
and therefore $\eval_\Hy(AJ) = \eval_\Hy(AJ')$.
\end{proof}

%


\providecommand{\bysame}{\leavevmode\hbox to3em{\hrulefill}\thinspace}
\providecommand{\MR}{\relax\ifhmode\unskip\space\fi MR }
\providecommand{\MRhref}[2]{%
  \href{http://www.ams.org/mathscinet-getitem?mr=#1}{#2}
}
\providecommand{\href}[2]{#2}


\end{document}